\documentclass{article}
\usepackage{caption}

\usepackage[utf8]{inputenc}
\usepackage{physics}

\usepackage[backend=bibtex, style = alphabetic, doi = false, url = true, isbn = false, maxbibnames = 6]{biblatex}
\bibliography{jointBibQ}
\renewbibmacro{in:}{}      
\newbibmacro{string+doi}[1]{\iffieldundef{doi}{#1}{\href{https://dx.doi.org/\thefield{doi}}{#1}}}
\DeclareFieldFormat{title}{\usebibmacro{string+doi}{\mkbibemph{#1}}}
\DeclareFieldFormat[article]{title}{\usebibmacro{string+doi}{\mkbibquote{#1}}}
\DeclareFieldFormat[incollection]{title}{\usebibmacro{string+doi}{\mkbibquote{#1}}}                   
\DeclareFieldFormat[inproceedings]{title}{\usebibmacro{string+doi}{\mkbibquote{#1}}}     

\usepackage{subcaption}

\usepackage{amssymb}
\usepackage{amsfonts}
\usepackage{amsmath}
\usepackage{graphicx}
\usepackage{indentfirst}
\usepackage{dsfont}

\usepackage{mathtools}

\usepackage{enumerate}
\usepackage{url}
\usepackage{lmodern}
\usepackage{newpxtext}

\usepackage{array}
\usepackage{multirow}
\usepackage{longtable}
\usepackage{algorithm}
\usepackage[noend]{algpseudocode}

\usepackage{amsthm}

\newtheorem{theorem}{Theorem}

\newtheorem{proposition}[theorem]{Proposition}
\newtheorem{lemma}[theorem]{Lemma}

\usepackage[symbol]{footmisc}

\usepackage{tikz}
\usetikzlibrary{calc}
\usetikzlibrary{positioning}
\usetikzlibrary{intersections}
\pgfdeclarelayer{bg}    
\pgfsetlayers{bg,main}  
\definecolor{darkblue}{rgb}{0.15,0.35,0.55}
\definecolor{reddish}{rgb}{.8, 0.2, 0.2}
\usepackage[pdfpagelabels, linktocpage=true]{hyperref}
\hypersetup{
colorlinks=true,
citecolor=darkblue,
linkcolor=reddish,
urlcolor=darkblue,
pdfauthor={},
pdftitle={},
pdfsubject={}
}

\usepackage{bbold}

\long\def\ca#1\cb{} 

\newcommand{\becs}{\begin{cases}}
\newcommand{\bem}{\begin{matrix}}

\newcommand{\dya}[1]{|#1\rangle\langle#1|}

\newcommand{\encs}{\end{cases}}
\newcommand{\enm}{\end{matrix}}

\newcommand{\inp}[1]{\langle#1|#1\rangle }

\newcommand{\ot}{\otimes }

\newcommand{\AC}{{\mathcal A}}

\newcommand{\HC}{{\mathcal H}}

\newcommand{\TC}{{\mathcal T}}

\newcommand{\Cbb}{\mathbb{C}}

\newcommand{\al}{\alpha }

\newcommand{\dl}{\delta }

\newcommand{\ep}{\epsilon}

\newcommand{\cone}{\text{cone}}
\newcommand{\conv}{\text{conv}}
\newcommand{\Herm}{\text{Herm}}
\newcommand{\EXT}{\text{EXT}}
\newcommand{\DPS}{\text{DPS}}
\newcommand{\SEP}{\text{SEP}}
\newcommand{\Sym}{\text{Sym}}

\newcommand{\FOMupdate}{\text{\tt FOMupdate}}
\newcommand{\FOMupdatepst}{\text{\tt FOMupdatepst}}

\def\outl#1{\par{\medskip\noindent\hspace*{0.1cm}\bf
      \mathversion{bold}#1\mathversion{normal}\smallskip} }
   \def\xa{} \def\xb{}  

\ca
 \def\outl#1{}\def\xa{}\def\xb{}
\cb

\ca
 \def\outl#1{\par{\medskip\noindent\hspace*{.5cm}\bf
      \mathversion{bold}#1\mathversion{normal}\smallskip} }
 \long\def\xa#1\xb{} 
  
\cb

\usepackage{booktabs}
\usepackage{graphicx} 
\usepackage{amssymb}
\usepackage{amsmath}
\usepackage{comment}
\usepackage{graphicx}
\usepackage{longtable}
\usepackage{subcaption}
\usepackage{multirow}
\usepackage{setspace}
\usepackage{geometry}
\usepackage{pdflscape}
\usepackage{booktabs}
\usepackage{graphics}
\usepackage{tikz}
\usepackage{url}
\usepackage{pgfplots}
\usepackage{algorithm}
\usepackage{algpseudocode}

\DeclareMathAlphabet{\pazocal}{OMS}{zplm}{m}{n}

\newcommand{\I}{\mathcal{I}}
\newcommand{\HSD}{\text{HSD}}
\newcommand{\PST}{\text{PST}}

\newcommand{\gap}{\text{gap}}

\newcommand{\blue}[1]{{#1}}
\newcommand{\D}{\mathcal D}

\usepackage[affil-it]{authblk}

\title {First-order and interior-point methods for
entanglement detection}
\author[1]{Javier Pe\~na\thanks{jfp@andrew.cmu.edu}}
\author[2]{Vikesh Siddhu\thanks{vsiddhu@protonmail.com}}
\author[3]{Sridhar Tayur\thanks{stayur@cmu.edu}}
\affil[1]{Tepper School of Business, Carnegie Mellon University}
\affil[2]{IBM Quantum, IBM T.J. Watson Research Center}
\affil[3]{Quantum Technologies Group, Tepper School of Business, Carnegie Mellon University}

\begin{document}
\maketitle

\begin{abstract}
   Quantum entanglement lies at the heart of quantum information science, yet its reliable detection in high-dimensional or noisy systems remains a fundamental computational challenge. Semidefinite programming (SDP) hierarchies, such as the Doherty-Parrilo-Spedalieri (DPS) and Extension (EXT) hierarchies, offer complete methods for entanglement detection, but it is well known that their practical use is limited by exponential growth in problem size if implemented naively. We make three contributions. First, we introduce a new SDP hierarchy, PST, that is sandwiched between EXT and DPS—offering a tighter approximation to the set of separable states than EXT, while incurring significantly lower computational overhead than DPS. Second, we explicitly construct compact, polynomially-scalable descriptions of EXT and PST using partition mappings and operators. These descriptions in turn yield formulations that satisfy desirable properties such as the Slater condition and are well-suited to both first-order methods (FOMs) and interior-point methods (IPMs). Third, we design a suite of entanglement detection algorithms: three FOMs (Frank-Wolfe, projected gradient, and fast projected gradient) based on a least-squares formulation, and a custom primal-dual IPM based on a conic programming formulation. These methods are numerically stable and capable of producing entanglement witnesses or proximity measures, even in cases where states lie near the boundary of separability.  Numerical experiments on benchmark quantum states demonstrate that our algorithms improve the ability to solve deeper levels of the SDP hierarchy. 

\end{abstract}

\newpage

\section{Introduction}
\label{sec:Intro}
Quantum entanglement is a cornerstone of modern quantum technologies that include
quantum computation, communication, cryptography, and
sensing~\cite{HorodeckiEA09}. Detection and characterization of
entanglement are essential for both foundational research and practical
applications in quantum information science (QIS); however, entanglement detection
remains a profound challenge, especially in large, noisy, or high-dimensional
systems~\cite{OtfriedGeza09, SiddhuTayur22}. The crux of this entanglement detection problem
lies in characterizing the convex set of separable~(un-entangled) states, {\tt
Sep}, and in the exponential growth of the Hilbert space.

To address these challenges, the QIS community has turned to convex
optimization methods based on semidefinite programming (SDP) relaxations of the
intractable set of separable states~\cite{DohertyParriloEA04, EisertEA04,
HarrowEA17}. An increasingly tight sequence of SDP relaxations that converges
to {\tt Sep} is called a {\em complete} SDP {\em hierarchy}. 
Explicit examples of these hierarchies come from foundational work of
Doherty-Parrilo-Spedalieri (DPS)~\cite{DohertyParriloEA04} which contains two
complete hierarchies, EXT and its tighter version DPS.  Both are connected to
the broader family of sum-of-squares (SOS) hierarchies \blue{and the method of
moments}~\cite{FangFawzi20}, such as the Lasserre/Parrilo
hierarchy~\cite{Lasserre01, Parrilo2003}.  \blue{Indeed, as detailed
in~\cite{gribling2022bounding}, the DPS hierarchy can also be derived via the
method of moments.}
Complete SDP hierarchies provide a systematic way to detect entanglement at the
cost of increased computational resources in the following way.  At each level
$k$ of the hierarchy an SDP tests a \blue{necessary} condition for a state to be
separable~(increasing $k$ generally increases the SDP's computational
complexity). While a separable state passes the test for all levels $k$, an
entangled state must fail at some level $k$ and then an explicit entanglement
witness can be constructed from the dual SDP.  


Solving SDPs is computationally taxing.
In the particular case of SDPs in hierarchies for entanglement
detection two additional difficulties arise: first, the size of these SDPs may
grow exponentially with level of the hierarchy; and second, off-the-shelf
solvers can take more iterations than necessary to detect entanglement.  To
realize the full potential of SDP hierarchy based approaches to entanglement
detection, it is likely necessary to find both tailored solutions for prior SDP
hierarchies and new hierarchies that admit a tailored solution.

Our work refines the SDP based approaches to entanglement detection, seeking to
balance computational efficiency with detection power. We do so through the
following four main developments (more details in Section~\ref{sec.results}
below): First, we introduce a new \textit{PST hierarchy} that is sandwiched
between the EXT and DPS hierarchies.  \blue{The PST hierarchy is naturally
motivated by computational considerations: While DPS demands a cascade of
increasingly large positive semidefinite constraints, PST adds only a single
constraint of the same size as that required in an efficient description of the
EXT hierarchy.  As numerical experiments demonstrate, at any level $k$,
$\PST_k$ is much tighter than $\EXT_k$ but far more efficient than $\DPS_k$,
thus making it possible to solve $\PST_k$ at significantly higher values
of $k$ if needed. Indeed, there are entangled states not detectable by
$\DPS_2$ but easily detectable via $\PST_3$.  Second, we give an explicit
construction of a new \textit{partition-based operator} $\AC$ to give new
explicit and efficient descriptions of the EXT, PST, and DPS hierarchies.
Although the existence of efficient descriptions for EXT and DPS has been
mentioned in the literature~\cite{DohertyParriloEA04,NavascuMasakiPlenio09,gribling2022bounding}),
no explicit constructions for them appear to have been documented prior
to this work.  The operator $\AC$ and the efficient descriptions of EXT, PST,
and DPS may be of independent interest beyond the entanglement detection
application of this paper.} Third, we propose two formulations - least squares
and conic programming - tailored to First Order Methods (FOMs) and Interior
Point Methods (IPMs) respectively, to harness the convergence properties of the
EXT and PST hierarchies. Fourth, we implement and test our algorithms for
entanglement detection on a collection of instances from the literature.

\subsection{Literature review and motivation}

A quantum state $\rho_{ab}$ on systems $a$ and $b$ with finite dimension $d_a$
and $d_b$, respectively, is separable~($\rho_{ab} \in \text{\tt Sep}$) if there
exist pure states $\{\ket{\chi_i}_a\}, \{ \ket{\phi_i}_b\}$ and a probability
distribution $\{q_i\}$ such that
\begin{equation}
     \rho_{ab} = \sum_i q_i \dya{\chi_i} \ot \dya{\phi_i}.
     \label{eq:cvxDec}
\end{equation}
Otherwise $\rho_{ab}$ is {\em entangled} and there is a separating hyperplane
between $\rho_{ab}$ and {\tt Sep}, called a {\em witness} of entanglement.  
The focus of this work is to detect if $\rho_{ab} \in \text{\tt Sep}$ or not
using tools from Mathematical Physics and Operations Research. 
A variety of mathematical tests to detect entanglement have been, and continue
to be developed~(see~\cite{HorodeckiEA09, OtfriedGeza09} and follow-up works) .
Of particular interest to us are the PPT test and some complete SDP hierarchies
based on the symmetric extension criterion (see Sec.~\ref{subsec.EXT}
and~\ref{subsec.DPS} for more details)~\cite{Peres96, HorodeckiHorodeckiEA96,
DohertyParriloEA04}.
These provide conclusive mathematical tests for entanglement and can be
performed numerically via SDPs.

While solving an SDP for successive levels of an SDP hierarchy gives an
algorithm that can eventually detect entanglement, in practice, there is a
maximum level $k$ up to which we can computationally solve the SDP hierarchy.
If such a computation claims to detect entanglement, then one may demand a
witness of entanglement, else entanglement is not detected upto level $k$ and
it is natural to ask how far the state may be from {\tt Sep} as a function of
$k$. In general, bounds on these distances with $k$ need not be
available~\cite{Beigi_2010}, however if a state passes the EXT test at level
$k$, $\EXT_k$, then its distance from {\tt Sep} scales as
$O(1/k)$~\cite{Ioannou07}.  A similar~(quadratically tighter) bound is known
for the DPS hierarchy~\cite{NavascuMasakiPlenio09}.  Thus increasing $k$ in EXT
and DPS gives a tighter bound on distance from {\tt Sep}.
Unfortunately, increasing the level $k$ in various SDP hierarchies or studying
systems with larger sizes of $d_a$~($d_b$) has thus far been numerically
challenging~\cite{OtfriedGeza09}. 
\blue{To the best of our knowledge, there does not seem to be any documented implementations of the DPS hierarchy for $k\ge 4$.  Some of our numerical experiments illustrate the magnitude of this challenge for $k=3$.} 

The numerical challenge can appear from at least three places. First is
translating an SDP from theory to implementation without increasing its size.
Second is formulating the SDP so it doesn't violate conditions under which it
can be solved efficiently. Third consideration is numerical, general purpose
convex optimization and semi-definite programming software~(PICOS
interface~\cite{Sagnol2022}, MOSEK~\cite{Mosek25}, and other solver) used
widely in quantum information science (QIS)~(for a recent solver tailored to
QIS see~\cite{HeSaundersonFawzi25a}) can become sluggist for even modestly
sized SDPs.


In general, a {\em well-behaved} SDP in standard form minimizing a linear
objective of $n \times n$ positive semidefinite matrices subject to $m$
equality constraints can be solved via interior-point methods~(IPM) to an
accuracy $\dl$ (for both optimality and feasibility) in $O(\sqrt{n}
\log(1/\dl))$ iterations.  Each interior-point iteration involves
$O(mn^3+m^2n^2+m^3)$ floating point operations~\cite{BenTal2021,ToddTT98}. In
massive-scale problems these iterations can become prohibitively costly and
first-order methods~(FOM) provide a viable alternative path. When minimizing a
convex function over {\em simple} domains, such as the spectraplex, a
FOM generates a solution within accuracy $\dl$ in  $O(1/\dl)$ or
$O(1/\dl^2)$ iterations depending on the particular structure of the problem.
Each iteration of a FOM may involve $O(mn^2)$ or $O(mn^2+ n^3)$
floating point operations depending on the particular
algorithm~\cite{BenTal2021}.
%


The worst-case complexity of algorithmic approaches to detect entanglement has
been widely studied with the aid of the weak membership problem: roughly
speaking, given a fixed state $\rho_{ab}$ and $\ep > 0$ decide if the
$\ep$-neighbourhood~(say in trace-distance) of $\rho_{ab}$ contains a separable
state.  This problem has been shown to be NP-hard in several settings and is
related to a variety of other
problems~(see~\cite{Gurvits04,HarrowEA17}).  This problem is
theoretically equivalent to optimizing certain linear functionals over {\tt
Sep}. This equivalence uses the ellipsoid method, which despite its excellent
theoretical properties is not perceived as a viable numerical tool.


It is of interest to estimate upper bounds on the runtime of solving the SDPs
appearing in a complete hierarchy for entanglement detection.  Let $\blue{d_{k}}:=
{d_b + k - 1 \choose k }$ be the dimension of the symmetric subspace over
$k$-copies of the $b$ system. Following~\cite{VandenbergheBoyd1996,
DohertyParriloEA04}, the analysis in~\cite{NavascuMasakiPlenio09} finds the
runtime of solving the SDP arising at the $k^{\text{th}}$ level of the EXT and
DPS hierarchies to be $O(d_a^6 \blue{d_{k}}^6)$ and $O(d_a^6 \blue{d_{k}}^4
\blue{d_{k/2}}^4)$, respectively. 
Using the convergence bounds relating $k$ to distance $\ep$ from {\tt Sep}
discussed previously, the time complexity to solve the weak membership problem
to error $\ep$ using EXT and DPS become $O\big(d_a^6 (2e/\ep)^{6 d_b} \big)$
and $O\big(d_a^6 (e^2/\ep)^{4 d_b} \big)$, respectively. \blue{Our efficient descriptions of $\EXT$ and $\PST$ yield the sharper bound $O\big(d_a^6 (2e/\ep)^{3 d_b} \big)$ for both $\EXT$ and $\PST$.  However, our numerical experiments demonstrate that PST performs vastly better than EXT.}

\subsection{Our results and methods}
\label{sec.results}
%

This work contributes to the goal of efficiently detecting entanglement. 
\begin{itemize}
    \item We propose a new SDP hierarchy, PST, sandwiched between EXT and DPS.
     \item \blue{We construct a partition-based operator $\AC$ to give explicit formulations of EXT, and PST whose size grows
        polynomially with the level of the hierarchy, as opposed to the
        exponential growth of a naive description.  The operator $\AC$ can 
        also be leveraged to give a formulation of DPS that is more efficient than those previously documented in the literature.}

    \item To solve EXT and PST we pursue two approaches:  a least-squares approach, and a
        conic-programming approach. Each approach is shown to be well-scaled and to satisfy the
        Slater condition. 
        %
        The least-squares approach is amenable to first-order methods (FOM) as
        it involves a least-squares objective function over a spectraplex. 
        %
        Our conic-programming approach is naturally amenable to interior-point
        methods (IPM) \blue{and thus has a straightforward implementation in PICOS.}  We also develop a custom
        IPM with computational advantages over generic solvers. 
        %
     
    \item We test our Python implementation of FOM and IPM entanglement
        detection methods on a collection of states from the literature. 
        Our numerical experiments illustrates the nuances that arise for states that are near the boundary of the set of separable states.
        
        %
\end{itemize}

\section{Preliminaries}
We use $\HC$ to denote the $d$-dimensional complex linear vector space
$\Cbb^d$. 
An element $\ket{\psi} \in \HC$ with unit norm,
$\inp{\psi} = 1$, represents a pure quantum state. 
The tensor product~(also called Kronecker product) of two spaces $\HC_a$ and
$\HC_b$ is $\HC_{ab} := \HC_a \ot \HC_b$. 
A bi-partite state $\ket{\psi}_{ab} \in\HC_{ab}$
 is called a product state when there are pure states $\ket{\phi}_a$ and
$\ket{\chi}_b$ such that $\ket{\psi}_{ab} = \ket{\phi}_a \ot \ket{\chi}_b$. Let $P_{ab} \subset
\HC_{ab}$ denote the set of pure product states. Pure states not in $P_{ab}$
are called entangled pure states.

Let $\Herm(\HC)$ represent the space of linear operators on $\HC$ that are
Hermitian. Any $\rho \in \Herm(\HC)$ is called a density operator when $\rho$
is positive semi-definite, $\rho \succeq 0$, and has unit trace, $\Tr(\rho) =
1$.  Let $\phi:=\dya{\phi}$, denote the projector onto pure state $\ket{\phi}$ 
and let $\D(\HC) \subseteq \Herm(\HC)$ denote the set of density operators in $\Herm(\HC)$.

We denote the convex hull of some set $A$ by
$\conv(A):= \{ \sum_i p_i a_i \; | \; a_i \in A, p_i \geq 0, \sum_i p_i = 1\}$ 
and denote the convex conic hull of $A$ by $\cone(A) := \{ \sum_i w_i a_i \; | \; a_i \in A, w_i \geq 0 \}$. Observe 
\begin{equation}
\D(\HC)  = \conv\{ \phi \; | \;  \ket{\phi} \in \HC, \; \inp{\phi} = 1\}. 
\end{equation}
%
In similar vein, the set of separable states on $\HC_{ab}$, is
\begin{equation}
\text{\tt Sep} :=    \conv\{ \phi_{ab}  \; | \; \ket{\phi}_{ab} \in P_{ab}\}.
\end{equation}
By construction, this set of separable states is a convex subset of $\D(\HC)$. Density operators not in this set are called entangled.
It is useful to consider the conic hulls of sets discussed thus far. In particular,
the cone of density operators is precisely the cone of positive semidefinite operators: 
\begin{equation}
\Herm^+(\HC):= \{\phi \in \Herm(\HC) \; | \;  \phi \succeq 0\} =\cone(\D(\HC)) = \{ c \rho \; | \; c \geq 0, \; \rho \in \D(\HC) \}. 
    \label{eq:densityOpCone}
\end{equation}
The convex cone of separable operators is denoted by
\begin{equation}
    \text{SEP} = 
    \cone(\text{\tt Sep}) 
= \{ c \rho \; | \; c \geq 0, \; \rho \in \text{\tt Sep}\}.
\end{equation}
%


The {\em trace} inner product~(also called the  Frobenius inner product) between $\rho,\sigma\in \Herm(\HC)$ takes the form
%
$    \sigma\bullet \rho  = \Tr(\sigma\rho ).    
$
%
Our developments rely on the following {\em polar} construction from convex analysis. 
%
Let $K \subseteq \Herm(\HC)$ be a convex cone. Then its polar is
\begin{equation}
    K^\circ = \{ \sigma \:|\; \sigma\bullet \rho \le 0 \; \forall \rho \in K\}.
    \label{eq:polar}
\end{equation}
It is well known and easy to see that whenever $K\subseteq \Herm(\HC)$ is a closed convex cone the following equivalence holds: 
$ 
\rho \in K \Leftrightarrow \sigma\bullet\rho \le 0 \text{ for all } \sigma \in K^\circ.
$ 

In particular, an operator $\rho \in \D(\HC_{ab})$ is entangled if and only if there exists $W \in \text{SEP}^\circ$ such that $\rho \bullet W > 0$.  When this is the case we say that $W$ is a {\em entanglement witness} for $\rho$.  Thus, the computational task of detecting whether an operator $\rho \in \D(\HC_{ab})$ is entangled can be \blue{equivalently} stated as that of finding an entanglement witness for $\rho$.

\section{SDP hierarchies: EXT, DPS, and PST}
\label{sec.hierarchies}

We next recall the SDP hierarchies EXT and DPS that approximate SEP.
EXT has a more efficient and simpler description whereas DPS converges faster to SEP.
We highlight  explicit  descriptions for EXT \blue{and DPS} that in turn suggest a new SDP hierarchy PST sandwiched between EXT and DPS.  \blue{In  contrast to DPS whose description is substantially more involved than that of EXT, the PST hierarchy has a description that is nearly as efficient as that of EXT.}

\subsection{EXT hierarchy}
\label{subsec.EXT}
For
any integer $k \geq 1$, let 
$    \HC_B :=  \HC_{b}^{\otimes k}$%
and let $\Tr_{b_{2:k}}:\HC_{aB} \rightarrow \HC_{ab}$ be the partial trace over the last $(k-1)$ components of $B$.   A density operator $\rho_{aB} \in \D(\HC_{aB})$
is an {\em extension} of $\rho_{ab}$ if $\rho_{aB} = \Tr_{b_{2:k}}(\rho_{aB}).$
Let  $\Sym(\HC_{B})\subseteq\HC_B$ denote the following subspace of {\em symmetric vectors}:
\[
\Sym(\HC_{B}):=\{\ket{\psi} \in \HC_{B} \;|\;\pi \cdot \ket{\psi} = \ket{\psi} \text{ for all } \pi\in \mathfrak{S}_k\}
\]
where $\mathfrak{S}_k$ is the symmetric group of $k$ elements and each $\pi \in \mathfrak{S}_k$ acts on $\HC_B$ by permuting the $k$ components:
$
\pi \cdot \ket{\psi_1} \ot \cdots \ot \ket{\psi_k} = \ket{\psi_{\pi(1)}} \ot \cdots \ot \ket{\psi_{\pi(k)}}.
$

Let $\Pi:\HC_B \rightarrow \HC_{B}$ denote the orthogonal projection onto $\Sym(\HC_{B})$. 
An extension $\rho_{aB}$ of $\rho_{ab}$ is called symmetric if 
$
    \rho_{aB} = (I_a\ot \Pi) \rho_{aB} (I_a \ot \Pi).
$
A quantum state $\rho_{ab}$ is separable if and only if it has a symmetric
extension for all $k \geq 1$~\cite{Werner89, CavesFuchsSchack02}.  Consider the
hierarchy of cones
\[
\EXT_k := \{\Tr_{b_{2:k}}(\rho_{aB}) \;|\; \rho_{aB} = (I_a\ot \Pi) \rho_{aB} (I_a \ot \Pi), \rho_{aB}\succeq 0\}.
\]
The $\EXT$ hierarchy of cones satisfies the following two properties:
\[
\EXT_{k+1} \subseteq \EXT_{k} \text{ for } k=1,2,\dots
\; \text{ and } \;
    \bigcap_{k=1}^\infty \EXT_{k} =  \text{SEP}.
\]
In particular, $\rho \not \in \SEP$ if $\rho \not \in \EXT_k$.  The latter in turn holds if and only if
there exists $W\in \EXT_k^\circ$ such that $\rho\bullet W >0$. 
We will rely on the following alternate description of $\EXT_k$.

\begin{proposition}\label{prop.ext}
Suppose $P:\HC\rightarrow \HC_B$ is a linear operator such that
$P(\HC) = \Sym(\HC_B)$. Then
\begin{equation}\label{alt.ext}
\EXT_k =
\{\Tr_{b_{2:k}}((I_a\ot P) X (I_a\ot P^\dagger)) \;|\; X\in \Herm^+(\HC_a \ot \HC)\}.
\end{equation}
\end{proposition}

\subsection{DPS hierarchy}
\label{subsec.DPS}
Another criterion for checking if a state $\rho_{ab}$ is separable is the
positive under partial transpose~(PPT) criterion.  Let 
$\TC_b:\Herm(\HC_{ab}) \rightarrow \Herm(\HC_{ab})$ denote the partial transpose with respect to $\HC_b$, that is, 
$\TC_b(\rho_a\ot\rho_b) = \rho_a \ot \rho_b^T$.
 For any separable state $\rho_{ab}$, the operator $\TC_{b}
(\rho_{ab})$ is positive semi-definite; thus, one says any separable
state $\rho_{ab}$ is PPT.  It thus follows that $\rho_{ab}$ is entangled if
$\TC_{b} (\rho_{ab})$ is not PPT~\cite{Peres96, HorodeckiHorodeckiEA96}. For
$d_a =2$ and $d_b = 2,3$, this PPT criterion is both necessary and
sufficient~\cite{HorodeckiHorodeckiEA96, Woronowicz76}. In general, this PPT
criterion is necessary but not sufficient~\cite{Horodecki97}.  

The Doherty-Parrilo-Spedalieri (DPS) criterion~\cite{DohertyParriloEA04} is a combination of the above symmetric extension criterion and PPT criterion.  In this criterion, for any $k \geq 1$,
the symmetric extension $\rho_{aB}$ of $\rho_{ab}$ must also be PPT, where the
partial transpose is taken with respect to each of the $k$ spaces $\HC_{b_1}$,
$\HC_{b_1 b_2}$, $\dots$, $\HC_{b_1 b_2 \dots b_k}$; that is,
\[
\TC_{b_1, b_2, \dots, b_j} (\rho_{aB}) \succeq 0 \text{ for } j=1,\dots,k.
\]
For ease of notation, write the previous equation as
$\TC_{b_{1:j}} (\rho_{aB}) \succeq 0 \text{ for } j=1,\dots,k.
$
This criterion can be stated in terms of the following hierarchy of cones $\DPS_k, \; k=1,2,\dots$ 
\[
\{\Tr_{b_{2:k}}(\rho_{aB}) \;|\; \rho_{aB} = (I_a\ot \Pi) \rho_{aB} (I_a \ot \Pi), \rho_{aB} \succeq 0, \; \TC_{b_{1:j}} (\rho_{aB}) \succeq 0 \text{ for } j=1,\dots,k\}.
\]
We will rely on the following analogue of Proposition~\ref{prop.ext} that gives an alternate description of $\DPS_k$.
\begin{proposition}\label{prop.dps}
Suppose $P:\HC\rightarrow \HC_B$ is a linear operator such that
$P(\HC) = \Sym(\HC_B)$. Then
\begin{equation}\label{alt.dps}
\begin{aligned}
&\DPS_k = \\
&\{\Tr_{b_{2:k}}((I_a\ot P) X (I_a\ot P^\dagger)) \;|\; X \in \Herm^+(\HC_a \ot \HC),  \TC_{b_{1:j}} ((I_a\ot P) X (I_a\ot P^\dagger)) \succeq 0 \text{ for } j=1,\dots,k\}.
\end{aligned}
\end{equation}
\end{proposition}
By construction, for each $k=1,2,\dots$ we have $\SEP\subseteq \DPS_k \subseteq \EXT_k$.  
Therefore the hierarchy of cones $\DPS_k$ also satisfies
\[
\DPS_{k+1} \subseteq \DPS_{k} \text{ for } k=1,2,\dots
\;
\text{ and }
\;
    \bigcap_{k=1}^\infty \DPS_{k}  = \text{SEP}.
\]

\subsection{PST hierarchy}

We next introduce some new SDP hierarchy, namely the $\PST$ hierarchy that is sandwiched between $\EXT$ and $\DPS$.
To that end, consider the following relaxed version of $\DPS_k$ as described in~\eqref{alt.dps}:
\[
\{\Tr_{b_{2:k}}((I_a\ot P) X (I_a\ot P^\dagger)) \;|\; X \in \Herm^+(\HC_a \ot \HC), \TC_{b_{1:k}} ((I_a\ot P) X (I_a\ot P^\dagger)) \succeq 0\}.
\]
We shall denote this set $\PST_k$.
\blue{We have the following analogue of Proposition~\ref{prop.ext} and Proposition~\ref{prop.dps} that gives an alternative description of $\PST_k$ and motivates its choice among the many possible cones sandwiched between $\EXT_k$ and $\DPS_k$.
\begin{proposition}\label{prop.pst}
Suppose $P:\HC\rightarrow \HC_B$ is a linear operator such that
$P(\HC) = \Sym(\HC_B)$. Then
\begin{equation}\label{alt.pst}
\PST_k = \{\Tr_{b_{2:k}}((I_a\ot P) X (I_a\ot P^\dagger)) \;|\;  X \in \Herm^+(\HC_a \ot \HC), \; \TC(X)\succeq 0 \},
\end{equation} 
where $\TC(X)$ is the  partial transpose with respect to $\HC$, that is, for $M\in L(\HC_a), Y \in L(\HC)$
\[
\TC(M\otimes Y) = M\otimes Y^T.
\]
\end{proposition}
}

By construction, we have $\DPS_k\subseteq \PST_k\subseteq\EXT_k$.  The PST hierarchy is thus stronger than EXT but with a lower overhead than DPS. \blue{The partition mappings and operators described in Section~\ref{sec.partition} below show that the main computational overhead in the descriptions of EXT, PST, and DPS shifts to the positive semidefinite constraints needed in each case. The description~\eqref{alt.ext} of EXT requires only one positive semidefinite matrix $X \in \Herm^+(\HC_a \ot \HC)$  and the description~\eqref{alt.pst} of PST requires only an additional positive semidefinite matrix of the same size as $X$, namely $\TC(X) \in \Herm^+(\HC_a \ot \HC)$.  By contrast, in addition to $X \in \Herm^+(\HC_a \ot \HC)$, the description~\eqref{alt.dps} of DPS requires $k$ additional positive semidefinite matrices, most of them larger than $X$, for the constraints $\TC_{b_{1:k}} ((I_a\ot P) X (I_a\ot P^\dagger)) \succeq 0, \; j=1,\dots,k.$}

\subsection{Entanglement detection via SDP hierarchies}

Let $\HSD$ denote a generic hierarchy of decreasing SDP-representable cones $\HSD_k$ such that 
\[
\bigcap_{k=1}^\infty \HSD_{k} =  \text{SEP}.
\]
For example, $\HSD$ could be one of the $\EXT, \DPS,$ or $\PST$ hierarchies.

An {\em ideal entanglement detection algorithm} with unlimited time and unlimited computational resources to solve semidefinite programs could proceed as follows: determine whether $\rho \in \HSD_k$ for $k=2,3,\dots$.  If $\rho\not \in \HSD_k$ for some finite $k$ then find $W\in \HSD_k^\circ \subseteq \SEP^\circ$ such that $W\bullet \rho > 0$ otherwise conclude that $\rho$ is separable.

However, in reality there is only limited computational time, and solving a semidefinite program requires non-trivial numerical procedures subject to limitations on computation time, memory usage, and precision. The next sections describe implementable algorithms that achieve the following 
{\em actual entanglement detection algorithm:} for some finite $k$ either  find $W\in \HSD_k^\circ \subseteq \SEP^\circ$ such that $W\bullet \rho > 0$ or conclude that $\rho$ is near some $\tilde \rho \in\HSD_k$.  In the former case we detect entanglement and $W$ is an entanglement witness whereas in the latter case we do not detect entanglement but instead conclude that $\rho$ is near $\SEP$.

\section{Partition mappings and operators $\AC$ and $\AC^{\dagger}$}\label{sec.partition}

The descriptions~\eqref{alt.ext},~\eqref{alt.dps}, and~\eqref{alt.pst} of $\EXT_k, \DPS_k,$ and  $\PST_k$ are stated in terms \blue{of the} matrices \blue{$(I_a\ot P) X (I_a\ot P^\dagger)$} that grow exponentially on $k$.  
This section describes alternate descriptions of $\EXT_k$ and $\PST_k$ that grow only polynomially on $k$ via some suitable operator $\AC$ and its adjoint
$\AC^{\dagger}$.  
This section describes the gist of the construction of $\AC$ and $\AC^\dagger$.  We give additional details of this construction in the Appendix.

To ease notation, throughout this section we will often write $d$ as shorthand for $d_b = \dim(\HC_b)$.  Thus $\dim(\HC_B) = d^k$ and $\dim(\Sym(\HC_B)) = d_k:= {d+k-1 \choose k}$.  The main building block of $\AC$ is a suitable {\em partition mapping} $P:\HC\rightarrow \HC_B$ that satisfies $P(\HC) = \Sym(\HC_B)$ and $P^\dagger P = I_{\HC}$ for some space $\HC$ with $\dim(\HC) = d_k$.
Let $\AC:\Herm(\HC_a\ot \HC)\rightarrow \Herm(\HC_{ab})$ denote the operator 
\[
X\mapsto \Tr_{b_{2:k}}((I_a\ot P) X (I_a\ot P^\dagger)).
\]
It thus follows that the adjoint operator $\AC^{\dagger}: \Herm(\HC_{ab})\rightarrow\Herm(\HC_a\ot \HC)$ is
\[
W \mapsto (I_a\ot P^\dagger) \left(W \ot I_{b_{2:k}}\right)(I_a\ot P).
\]
The following proposition readily follows from the construction of $\AC$ and the fact that $P^\dagger P = I_{\HC}$.
\begin{proposition}\label{prop.prop.A} The operator $\AC:\Herm(\HC_a\ot \HC)\rightarrow \Herm(\HC_{ab})$ satisfies the following two properties.  First, 
 $\AC(X) \in \Herm^+(\HC_{ab})$ for all $X\in \Herm^+(\HC_a\ot \HC)$.  Second, $\AC^{\dagger}(I_{ab}) = I_{a}\ot I_{\HC}.$
\end{proposition}
To ease notation we will write $\I$ as shorthand for $I_{a}\ot I_{\HC}$ throughout the sequel and thus write the second property above as $\AC^\dagger(I_{ab}) = \I$.  We will rely on this property for our main developments.

A naive implementation of the operators $\AC$ and $\AC^\dagger$ would involve computing matrices of dimension $d^k \times d^k$ in some intermediate step.  However, that exponential blowup can be entirely circumvented.  The operator $\AC$, when reshaped as a matrix, is quite sparse. 
Indeed, the number of nonzero entries in each row of $\AC$ reshaped as matrix is $d_{k-1}$.  Thus the total number of nonzero entries in $\AC$ is  $d_a^2d_b^2d_{k-1}$ and they can be explicitly described as shown in~\eqref{eq.Mentries} in the Appendix. \blue{Although the operator $\AC$ depends on the exponentially large partition mapping $P$, our construction of $\AC$ avoids constructing $P$ entirely.}

Equation~\eqref{alt.ext} in 
Proposition~\ref{prop.ext} implies that the cones $\EXT_k$ and $\EXT_k^\circ$ can be described as follows:
\[
\EXT_k =
\{\AC(X) \;|\; X \succeq 0\} \quad \text{ and } \quad \EXT_k^\circ  
=\left\{W \; | \; \AC^{\dagger}(W)  \preceq 0 \right\}.
\]


Likewise, Equation~\eqref{alt.pst} implies that the cones $\PST_k$ and $\PST_k^\circ$ can be described as follows:
\[
\PST_k =
\{\AC(X) \;|\; X \succeq 0, \; \TC(X)\succeq 0 \}, \;  \PST_k^\circ  
=\left\{W \; | \; \AC^{\dagger}(W) + \TC^\dagger(Z)  \preceq 0 \text{ for some } Z\succeq 0\right\}.
\]

Suppose $\rho\in\D(\HC_{ab})$.  Observe that $\rho\in \EXT_k$ if and only if the following equivalent systems of constraints are  feasible
\begin{equation}\label{EXT.primal}
\begin{aligned}
      \AC(X)&= \rho  \\
     X &\succeq 0.
    \end{aligned}
    \qquad \Leftrightarrow \qquad
    \begin{aligned}
      \AC(X)&= \rho \\
      \I\bullet X &= 1\\
     X &\succeq 0
    \end{aligned}
    \qquad \Leftrightarrow \qquad
    \begin{aligned}
      \AC(X)&= \rho \\
     X &\in \D(\HC_a\ot\HC).
    \end{aligned}
\end{equation}
The equivalence in~\eqref{EXT.primal} holds because $\AC^{\dagger}(I_{ab}) = \I$ and $I_{ab}\bullet \rho = 1$.
Similarly, $\rho\in \text{PST}_k$ if and only if the following equivalent systems of constraints are  feasible
\begin{equation}\label{PST.primal}
\begin{aligned}
      \AC(X)&= \rho  \\
      X & \succeq 0 \\
     \TC(X) &\succeq 0
    \end{aligned}
    \qquad \Leftrightarrow \qquad
    \begin{aligned}
      \AC(X)&= \rho \\
     \TC(X) - Y &= 0 \\
     X,Y &\in \D(\HC_a\ot\HC).
    \end{aligned}
\end{equation}
We conclude this section with the following scaling properties of the operator $\AC$.  

\begin{proposition}\label{prop.norm}
The Euclidean norm of the operators 
 $\AC^\dagger\AC :\blue{\Herm}(\HC_{a}\ot \HC)\rightarrow \blue{\Herm}(\HC_{a}\ot \HC)$ and  
$\AC \AC^\dagger:\blue{\Herm}(\HC_{ab})\rightarrow \blue{\Herm}(\HC_{ab})$
is \[
      \|\AC^\dagger \AC\| = \|\AC \AC^\dagger\| = \frac{d_{k}}{d},
      \]
 and the Euclidean norm of the operator $(\AC \AC^\dagger)^{-1}:\blue{\Herm}(\HC_{ab})\rightarrow \blue{\Herm}(\HC_{ab})$ can be bounded as follows
    \[
      \|(\AC \AC^\dagger)^{-1} \| \le \frac{d(d+1)}{d_{k}}.
    \]
\end{proposition}


\section{Algorithms to detect entanglement based on $\EXT$}
\label{sec.EXT}

As we detailed in Section~\ref{sec.partition}, membership in $\EXT_k$ can be formulated as the semidefinite system of constraints~\eqref{EXT.primal}.  This section describes several implementable algorithms that achieve the following version of the {\em Actual Entanglement Detection Algorithm} described in Section~\ref{sec.hierarchies}.  For some finite $k$:
\begin{itemize}
    \item Solve an alternative of~\eqref{EXT.primal} to obtain $W\in \EXT_k^\circ \subseteq \SEP^\circ$ with $W\bullet \rho > 0$.   In this case we detect entanglement and $W$ is an entanglement witness.
    \item Solve~\eqref{EXT.primal} approximately and conclude that $\rho$ is near a point $\tilde \rho \in\EXT_k$ and hence may or may not be entangled.    
\end{itemize}

\subsection{Formulation amenable to first-order methods (FOM)}
\label{sec.fom.ext}
The third system of constraints in~\eqref{EXT.primal} suggests the following least-squares approach that is amenable to first-order methods (FOM). 
Observe that $\D(\HC_a\ot \HC)=\{X\in \Herm(\HC_a\ot \HC) \;|\; X\succeq 0, \I\bullet X=1\}$ is the ``spectraplex'' in $\Herm(\HC_a\ot \HC)$.  Consider the least-squares problem
\begin{equation}\label{eq.ls.primal}
    \begin{aligned}
      \min_X \; & \frac{1}{2}\|\AC(X)- \rho\|^2 \\
      & X \in \D(\HC_a\ot \HC)
    \end{aligned} \quad \Leftrightarrow        \quad 
    \begin{aligned}
      \min_{x,X} \, \max_{u} \; & \frac{1}{2}\|x\|^2 +  (\AC(X)- \rho - x)\bullet u \\
      & X \in \D(\HC_a\ot \HC)
    \end{aligned}
\end{equation}
and its Fenchel dual
\begin{equation}\label{eq.ls.dual}
    \begin{aligned}
    \max_{u} \, \min_{x,X}  \; & \frac{1}{2}\|x\|^2 +  (\AC(X)- \rho - x)\bullet u \\
      & X \in \D(\HC_a\ot \HC)
    \end{aligned}
\quad \Leftrightarrow     \quad  
    \begin{aligned}
      \max_u & -\frac{1}{2}\|u\|^2 - \rho\bullet u - \lambda_{\max}(-\AC^{\dagger}(u)).
    \end{aligned}
\end{equation}
It is evident that both~\eqref{eq.ls.primal} and~\eqref{eq.ls.dual} satisfy the Slater condition.  Thus both of them attain the same optimal value.
Observe that $\rho\not\in\EXT_k$ if and only if the problem~\eqref{eq.ls.primal} has positive optimal value.  When that is the case and $u$ is such that $-\rho\bullet u - \lambda_{\max}(-\AC^{\dagger}(u)) > 0$ the point $W:=-u -\lambda_{\max}(-\AC^{\dagger}(u)) I_{ab}$ is an entanglement witness as it satisfies both $\rho\bullet W>0$ and $\lambda_{\max}(\AC^{\dagger}(W))\le 0$.

For $X\in \D(\HC_a\ot \HC)$ and $u\in \Herm(\HC_{ab})$ let $\gap(X,u)$ denote the {\em duality gap} between~\eqref{eq.ls.primal} and ~\eqref{eq.ls.dual}, that is,
\[
\gap(X,u):= \frac{1}{2}\|\AC(X)- \rho\|^2 + \frac{1}{2}\|u\|^2 + \rho\bullet u + \lambda_{\max}(-\AC^{\dagger}(u)).
\]
By convex duality, we have $\gap(X,u)\ge 0$ for all $X\in \D(\HC_a\ot \HC)$ and $u\in \Herm(\HC_{ab})$ with equality precisely when both $X\in \D(\HC_a\ot \HC)$ and $u\in \Herm(\HC_{ab})$ are optimal solutions to~\eqref{eq.ls.primal} and~\eqref{eq.ls.dual} respectively.

To determine whether $\rho \in \EXT_k$ or $\rho \not \in \EXT_k$, we use a first-order method to generate
a sequence of primal-dual pairs $(X_t, u_t) \in \D(\HC_a\ot \HC) \times \Herm(\HC_{ab})$ such that
$$
\gap(X_t,u_t) = \frac{1}{2}\|\AC X_t-\rho\|^2 + \frac{1}{2}\|u_t\|^2 + \rho \bullet u_t + \lambda_{\max}(-\AC^{\dagger} u_t) \rightarrow 0.
$$
Since the first two terms of $\gap(X_t,u_t)$ are non-negative, we are guaranteed to come across one of the following two possibilities.
\begin{itemize}
\item {\bf Detect entanglement:} find an entanglement witness
$W:=-u_t -\lambda_{\max}(-\AC^{\dagger} u_t) I_{ab}$  if $$\rho \bullet u_t + \lambda_{\max}(-\AC^{\dagger} u_t)<0.$$
This case is as depicted in Figure~\ref{fig.ext}(a).

\item {\bf Proximity of $\rho$ to $\EXT_k$:} find $X_t\in \D(\HC_a\ot \HC)$ such that $\frac{1}{2}\|\AC X_t-\rho\|^2\le \gap(X_t,u_t)$ is small if both $\gap(X_t,u_t)$ is small and
$$\rho \bullet u_t + \lambda_{\max}(-\AC^{\dagger} u_t)\ge0.$$
In this case entanglement is not detected and instead we find a point  $\tilde \rho = \AC X_t\in \EXT_k$ near $\rho$.  This case is as depicted in Figure~\ref{fig.ext}(b).
\blue{Our numerical experiments in Section~\ref{sec.examples} illustrate the conceptual content of Figure~\ref{fig.ext}(a) and Figure~\ref{fig.ext}(b).}
\end{itemize}

\begin{figure}[!t]
 \hspace{-.6in}
\begin{tabular}{c c} 
\begin{subfigure}{.6\textwidth}
    \centering
        \includegraphics[width=.6\textwidth]{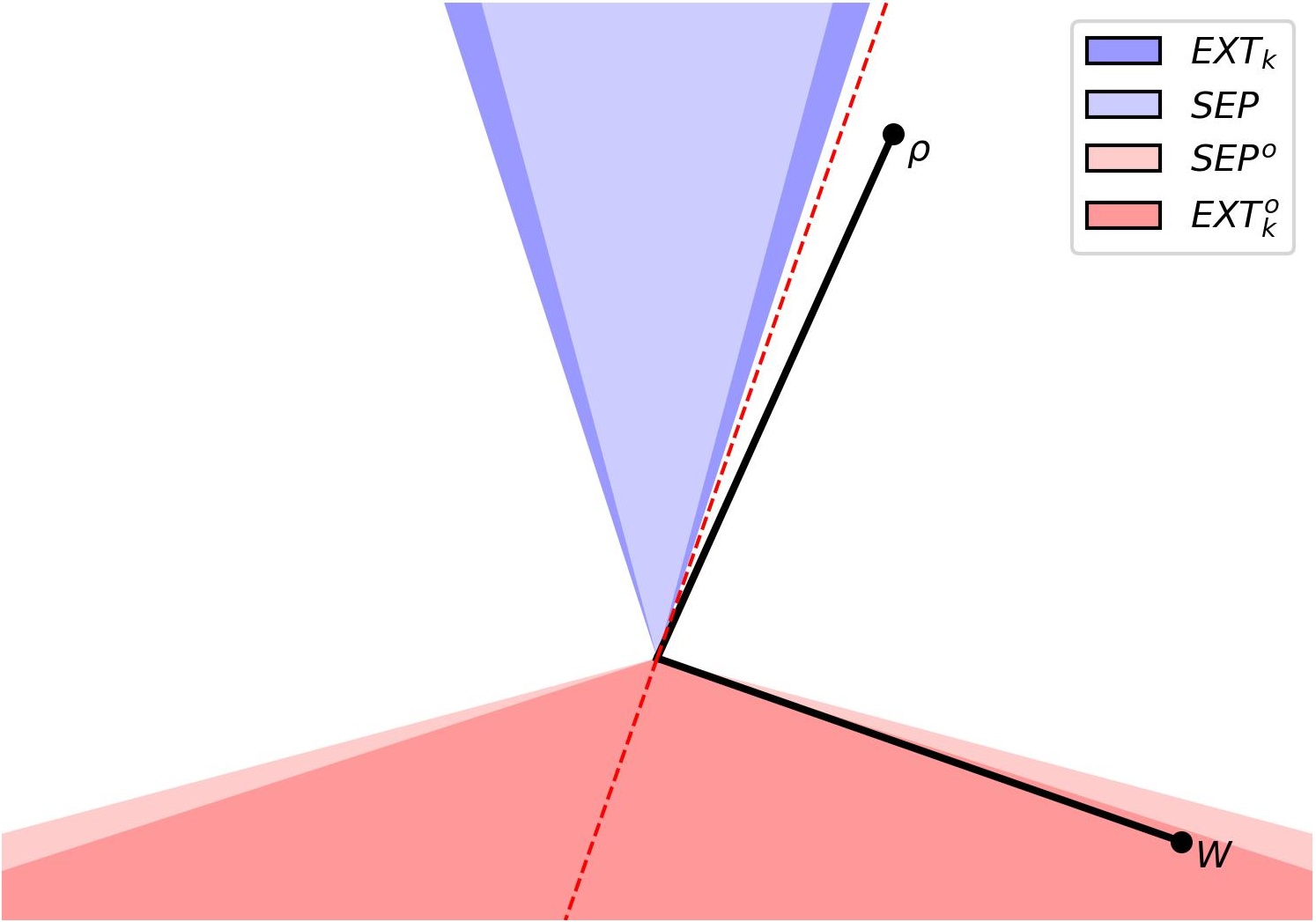}
        \caption{Entanglement detected for $\rho \not\in \EXT_k$}
    \end{subfigure} & \hspace{-.5in}
    \begin{subfigure}{.6\textwidth}
    \centering
        \includegraphics[width=.6\textwidth]{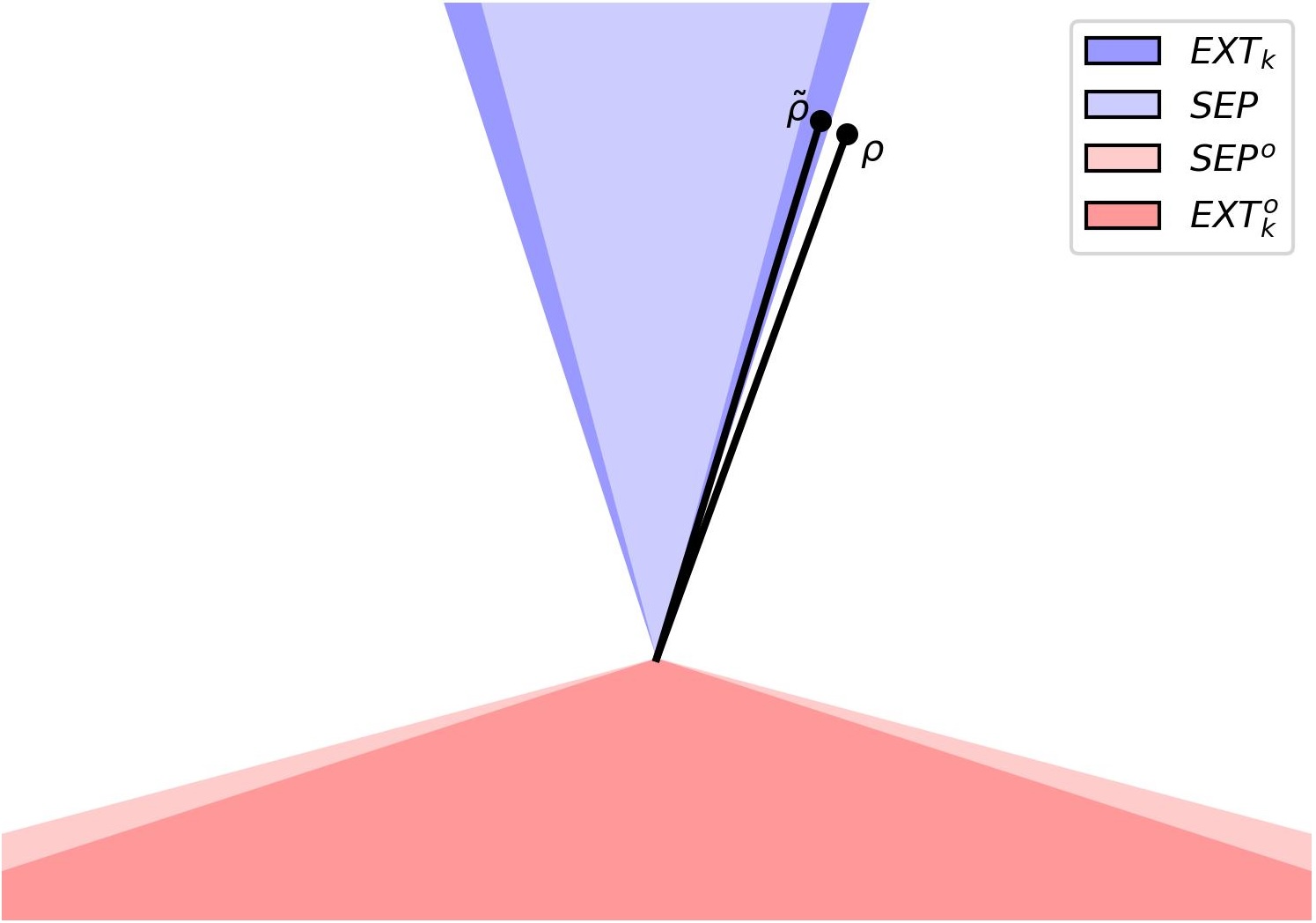}
        \caption{Proximity of $\rho$ to $\EXT_k$.}
    \end{subfigure} 
\end{tabular}
\caption{Entanglement detection via $\EXT_k$.}\label{fig.ext}
\end{figure}

Algorithm~\ref{algo.fom} describes a first-order algorithmic template to  generate a sequence $(X_t, u_t) \in \D(\HC_a\ot \HC) \times \Herm(\HC_{ab})$ such that $\gap(X_t,u_t) \rightarrow 0$.

\begin{algorithm}
    \caption{First-order algorithmic template for EXT}\label{algo.fom}
	\begin{algorithmic}        
    \State{\bf Input:} {$\rho \succeq 0$ with $I_{ab}\bullet \rho = 1$ and $X_0 = \frac{1}{\I\bullet \I}\I\in \D(\HC_a\ot \HC)$.
    }
    \State{\bf Output:} {$(X_t,u_t) \in \D(\HC_a\ot \HC)\times \Herm(H_{ab})$ such that $\gap(X_t,u_t) \rightarrow 0$}
    \For{$t=0,1,\dots,$}
    \State {$(u_{t},X_{t+1}):=\FOMupdate(t)$}
    \EndFor 
\end{algorithmic}    
\end{algorithm}
The main update in Algorithm~\ref{algo.fom}, namely, $(u_{t},X_{t+1}):=\FOMupdate(t),$ can be performed via one any of the following three iconic first-order algorithmic schemes as detailed in the appendix: Frank-Wolfe (FW),  Projected gradient (PG), and  fast projected gradient (FPG).  These schemes are viable because both a linear oracle and a projection oracle are easily computable for the constraint set in the least-squares problem~\eqref{eq.ls.primal}, namely $\D(\HC_a\ot \HC)$, \blue{since this set is a spectraplex~\cite{BenTal2021}.}
The following proposition, which is an immediate consequence of the scaling properties of $\AC$ and the properties of the Frank-Wolfe, projected gradient, and fast projected gradient methods as documented in~\cite{BenTal2021,Braun2022,GutmPena2023,Jaggi2013}.

\begin{proposition}\label{prop.fom}
The sequence $(X_t, u_t) \in \Delta \times \Herm(\HC_{ab})$ generated by Algorithm~\ref{algo.fom} satisfies
\begin{equation}\label{eq.fom.conv}
\gap(X_t, u_t) \le \left\{ \begin{array}{ll} \frac{4}{t+2} & \text{ for the FW update} \\ \frac{2d_k}{d_bt} & \text{ for the PG update}  \\ \frac{8d_k}{d_b(t+1)^2} & \text{ for the FPG update.} \end{array}\right.
\end{equation}
Each main iteration requires $O((d_ad_k)^2)$ floating point operations (flops) in addition to either a smallest eigenpair computation (for FW) or a full eigenvalue decomposition computation (for PG and FPG) of a Hermitian matrix of size $d_ad_k \times d_ad_k$.  Thus the number of flops per iteration is approximately
$O((d_ad_k)^2)$ for FW and $O((d_ad_k)^3)$ for PG and FPG.
\end{proposition}

\subsection{Formulation amenable to interior-point methods (IPM)}
\label{sec.ipm.ext}
The first system of constraints~\eqref{EXT.primal} also suggests the following conic programming approach that is amenable to interior-point methods (IPM).

Consider the conic programming problem
\begin{equation}\label{eq.cp.primal}
    \begin{aligned}
    \min_{X,\mu} \; &  \mu \\ 
        \text{subject to} \; & \AC X -\mu I = \rho \\
        & X \succeq 0
    \end{aligned}
\end{equation}
and its conic dual
\begin{equation}\label{eq.cp.dual}
    \begin{aligned}
    \max_{W} \; &  \rho\bullet W \\ 
        \text{subject to} \; & \AC^{\dagger} W \preceq 0 \\
        & -I\bullet W = 1.
    \end{aligned} \quad \Leftrightarrow \quad
        \begin{aligned}
    \max_{W,S} \; &  \rho\bullet W \\ 
        \text{subject to} \; & \AC^{\dagger} W + S = 0 \\
        & -I\bullet W = 1\\
        & S\succeq 0.
    \end{aligned}
\end{equation}
It is worth noting that the number of linear constraints (typically denoted $m$) in 
~\eqref{eq.cp.primal} is constant and equal to $d_a^2d_b^2$ regardless of the level $k$ of the hierarchy.  It is also evident that both~\eqref{eq.cp.primal} and~\eqref{eq.cp.dual} satisfy the Slater condition.  Thus both of them attain the same optimal value.
Observe that $\rho \in \EXT_k$ if and only if the primal problem has non-positive optimal value.  When that is not the case a dual feasible point $W$ is an entanglement witness if $\rho\bullet W > 0$.
Again to determine which of these two possibilities occurs, we use an interior-point method to generate a sequence of tuples $(X_t,\mu_t,W_t,S_t)$ such that $(X_t,\mu_t)$ is primal feasible, $(W_t,S_t)$ is dual feasible, and
\[
\gap(X_t,\mu_t,W_t,S_t):=
\mu_t - \rho\bullet W_t \rightarrow 0.
\]
We are guaranteed to come across one of the following two possibilities.
\begin{itemize}
\item {\bf Detect entanglement:} find an entanglement witness
$W_t$  if $\rho \bullet W_t > 0$.
\item {\bf Proximity of $\rho$ to $\EXT_k$:} find $X_t\succeq 0$ such that $\AC X_t -\rho = \mu_t I$ with $\mu_t \le \gap(X_t,\mu_t,W_t,S_t)$ is small if $\rho \bullet W_t \ge 0$.  In this case entanglement is not detected and instead we find a point $\AC X_t\in \EXT_k$ near $\rho$ (if $\mu_t > 0$) or find $\AC X \in \EXT_k$ (if $\mu_t \le 0$) such that $\AC X=\rho$.  
\end{itemize}

The primal-dual pair of conic programs~\eqref{eq.cp.primal}--\eqref{eq.cp.dual} can be readily implemented and solved with general-purpose conic programming solvers.  Since general-purpose solvers generate infeasible iterates, it is typically required that the pair~\eqref{eq.cp.primal}--\eqref{eq.cp.dual} be solved to full optimality before one of the above two possibilities is obtained.  Thus we propose a custom interior-point method solver for~\eqref{eq.cp.primal}--\eqref{eq.cp.dual} that generates {\em feasible} primal iterates $(X_t,\mu_t)$ and dual iterates $(W_t,S_t)$ such that $\gap(X_t,\mu_t,W_t,S_t)\rightarrow 0$.  An advantage of this algorithm is that it can detect entanglement early. The key step to ensure feasibility is to select an initial primal-dual feasible tuple $(X_0,\mu_0,W_0,S_0)$. To that end, select components $W_0,S_0$  as follows
\begin{equation}\label{init.dual}
W_0:=-\frac{1}{d_{ab}}I_{ab}, \;\;S_0 := -\AC^{\dagger}(W_0) = \frac{1}{d_{ab}}\I.
\end{equation}
It is evident that this pair $(W_0,S_0)$ is dual feasible with $S_0\succ 0$.
Select $(X_0,\mu_0)$ as follows
\begin{equation}\label{init.primal}
\bar X :=  \AC^\dagger(\AC\AC^\dagger )^{-1}\rho, \; \;\mu_0 := d_k, \; \;
X_0:= \bar X + \frac{\mu_0}{d_k} \I.
\end{equation}
Proposition~\ref{prop.norm} and $\AC(\I) = d_k I$ imply that $(X_0,\mu_0)$ is primal feasible with $X_0 \succ 0$.

The following proposition is an immediate consequence of standard properties of interior-point methods as detailed in~\cite{Renegar2001,BenTal2021}.

\begin{proposition}\label{prop.ipm}
The sequence $(X_t,\mu_t,W_t,S_t)$ generated by 
a feasible interior-point method applied to the primal-dual pair starting from the initial tuple primal-dual feasible tuple $(X_0,\mu_0,W_0,S_0)$ defined via~\eqref{init.dual} and ~\eqref{init.primal} generates $(X_t,\mu_t,W_t,S_t)$ such that $(X_t,\mu_t)$   is primal feasible, $(W_t,S_t)$ is dual feasible and
\begin{equation}\label{eq.ipm.conv}
\mu_t - \rho\bullet W_t =O
\left( 
\frac{\mu_0 - \rho\bullet W_0}{\exp(\sqrt{d_ad_k}t)}
\right) = O
\left( 
\frac{d_k}{\exp(\sqrt{d_ad_k}t)}
\right). 
\end{equation}
Each main iteration requires $O((d_ad_k)^3)$ floating point operations (flops).
\end{proposition}

\section{Algorithms to detect entanglement based on $\PST$}
\label{sec.PST}

This section provides parallel developments to those in Section~\ref{sec.EXT} but with the $\PST$ hierarchy in lieu of $\EXT$.  
As we detailed Section~\ref{sec.partition}, membership in $\PST_k$ can be formulated as the semidefinite system of constraints~\eqref{PST.primal}. 






\subsection{Formulation amenable to first-order methods (FOM)}
\label{sec.fom.pst}
Consider the constrained least-squares problem
\begin{equation}\label{eq.lspst.primal}
 \begin{aligned}
      \min_{X,Y} \; & \frac{1}{2}\|\AC(X)- \rho\|^2  + \frac{1}{2}\|\TC(X)-Y\|^2\\
	& (X,Y) \in \D(\HC_a\ot \HC)\times \D(\HC_a\ot \HC)
    \end{aligned}
\end{equation}
and its Fenchel dual
\begin{equation}\label{eq.lspst.dual}
    \begin{aligned}
      \max_{u,z} & -\frac{1}{2}\|u\|^2 - \frac{1}{2}\|z\|^2  - \rho\bullet u - \lambda_{\max}\left(-\AC^{\dagger}(u)-\TC^\dagger(z)\right) - \lambda_{\max}(z).
    \end{aligned}
\end{equation}
In light of~\eqref{PST.primal}, $\rho\not\in\PST_k$ if and only if the primal problem has positive optimal value.  Furthermore, if $u,z$ are such that $-\rho\bullet u - \lambda_{\max}(-\AC^{\dagger}(u) - \TC^\dagger(z)) - \lambda_{\max}(z) > 0$ then the component $W$ of the pair $(W,Z)$ defined as follows provides an entanglement witness:  
$$W:=-u -(\lambda_{\max}(-\AC^{\dagger} u-\TC^\dagger(z)) + \lambda_{\max}(z) ) I_{ab},\qquad  Z:= -z + \lambda_{\max}(z) \I.$$ 
As in Section~\ref{sec.EXT}, we use a first-order method to generate
a sequence of primal-dual tuples $(X_t, Y_t, u_t, z_t) \in \D(\HC_a\ot \HC)\times\D(\HC_a\ot \HC)\times\Herm(\HC_{ab})\times \Herm(\HC_a\ot \HC)$ such that $\gap(X_t, Y_t, u_t, z_t) \rightarrow 0.$
We are guaranteed to detect entanglement if $\rho \bullet u_t + \lambda_{\max}(-\AC^{\dagger} u_t - \TC^\dagger z_t)+\lambda_{\max}(z_t) < 0$
 or proximity to $\PST_k$ if both $\gap(X_t, Y_t, u_t, z_t)$ is small and
$$\rho \bullet u_t + \lambda_{\max}(-\AC^{\dagger} u_t - \TC^\dagger z_t)+\lambda_{\max}(z_t) \ge0.$$

\begin{algorithm}
    \caption{First-order algorithmic template for PST}\label{algo.fom.pst}
	\begin{algorithmic}     
    \State{\bf Input:} {$\rho \succeq 0$ with $I_{ab}\bullet \rho = 1$ and $X_0 = Y_0 = \frac{1}{\I\bullet \I}\I\in \D(\HC_a\ot \HC)$.
    }
    \State{\bf Output:} {$(X_t,Y_t,u_t,z_t) \in 
\D(\HC_a\ot \HC)\times \D(\HC_a\ot \HC)\times \Herm(H_{ab})\times \Herm(\HC_a\otimes \HC)$ such that $\gap(X_t,Y_t,u_t,z_t) \rightarrow 0$}
	
    \For{$t=0,1,\dots,$}
    \State {$(u_{t},z_t,X_{t+1},Y_{t+1}):=\FOMupdatepst(t)$}
    \EndFor 
\end{algorithmic}    
\end{algorithm}

Algorithm~\ref{algo.fom.pst} describes a first-order algorithmic template to  generate a sequence $(X_t, u_t) \in \D(\HC_a\ot \HC) \times \Herm(\HC_{ab})$ such that $\gap(X_t,u_t) \rightarrow 0$.  
The main update in Algorithm~\ref{algo.fom.pst}, namely, $$(u_{t},z_t,X_{t+1},Y_{t+1}):=\FOMupdatepst(t),$$ can be performed via the Frank-Wolfe, projected gradient, or fast projected gradient algorithmic schemes.  They are straightforward extensions of their counterparts for $\FOMupdate(t)$.  

The following analogue of Proposition~\ref{prop.fom} quantifies the convergence rate of $\gap(X_t,Y_t, u_t,z_t)\rightarrow 0$.

\begin{proposition}\label{prop.fom.pst}
The sequence $(X_t,Y_t, u_t,z_t) \in \D(\HC_a\ot \HC) \times \Herm(\HC_{ab})$ generated by Algorithm~\ref{algo.fom} satisfies
\begin{equation}\label{eq.fom.conv}
\gap(X_t,Y_t, u_t,z_t) \le \left\{ \begin{array}{ll} \frac{20}{t+2} & \text{ for the FW update} \\ \frac{8(d_k+2d_b)}{d_bt} & \text{ for the PG update}  \\ \frac{32(d_k+2d_b)}{d_b(t+1)^2} & \text{ for the FPG update.} \end{array}\right.
\end{equation}
Each main iteration requires 
$O((d_ad_k)^2)$ flops for FW and $O((d_ad_k)^3)$ for PG and FPG.
\end{proposition} 

\subsection{Formulation amenable to interior-point methods (IPM)}
\label{sec.ipm.pst}
The first system of constraints in~\eqref{PST.primal} also suggests the following conic programming approach that is amenable to interior-point methods (IPM).
Consider the primal-dual pair
\begin{equation}\label{eq.conic.pst}
    \begin{aligned}
    \min_{X,\mu} \; &  \mu \\ 
        \text{subject to} \; & \AC X -\mu I = \rho \\
        & X\succeq 0, \TC(X) \succeq 0,
    \end{aligned} \qquad \text{ and } \qquad
    \begin{aligned}
    \max_{W,S,Z} \; &  \rho\bullet W \\ 
        \text{subject to} \; & \AC^{\dagger} W + S +  \TC^\dagger(Z) = 0 \\
        & -I\bullet W = 1\\
        & S \succeq 0, Z \succeq 0.
    \end{aligned}
\end{equation}
Once again, both of these conic programs satisfy the Slater condition and thus attain the same optimal value.  We 
use a feasible interior-point method to generate a sequence of tuples $(X_t,\mu_t,W_t,S_t,Z_t)$ such that $(X_t,\mu_t)$ is primal feasible, $(W_t,S_t,Z_t)$ is dual feasible, and
\[
\gap(X_t,\mu_t,W_t,S_t,Z_t):=
\mu_t - \rho\bullet W_t \rightarrow 0.
\]
Again we are guaranteed to detect entanglement or proximity to $\PST_k$.
Select $(W_0,S_0,Z_0)$ as follows
\begin{equation}\label{init.dual.pst}
W_0:=-\frac{1}{d_{ab}}I_{ab}, \;\;S_0 := Z_0 := -\frac{1}{2}\AC^{\dagger}(W_0) = \frac{1}{2d_{ab}}\I.
\end{equation}
It is evident that this pair $(W_0,S_0,Z_0)$ is dual feasible with $S_0\succ 0,Z_0\succ 0$. Select $(X_0,\mu_0)$ as follows
\begin{equation}\label{init.primal.pst}
\bar X :=  \AC^\dagger(\AC\AC^\dagger )^{-1}\rho, \; \;\mu_0 := d_k, \; \;
X_0:= \bar X + \frac{\mu_0}{d_k} \I.
\end{equation}

Proposition~\ref{prop.norm} and $\AC(\I) = d_k I$ imply that $(X_0,\mu_0)$ is primal feasible with $X_0\succ 0, \TC(X_0) \succ 0$.

The following proposition, analogous to Proposition~\ref{prop.ipm}, is also an immediate consequence of standard properties of interior-point methods~\cite{Renegar2001,BenTal2021}.

\begin{proposition}\label{prop.ipm.pst}
a feasible interior-point method applied to the primal-dual pair starting from the initial tuple primal-dual feasible tuple $(X_0,\mu_0,W_0,S_0,Z_0)$ defined via~\eqref{init.dual.pst} and ~\eqref{init.primal.pst} generates a sequence $(X_t,\mu_t,W_t,S_t,Z_t)$  such that $(X_t,\mu_t)$   is primal feasible, $(W_t,S_t,Z_t)$ is dual feasible and
\begin{equation}\label{eq.ipm.conv}
\mu_t - \rho\bullet W_t =O
\left( 
\frac{\mu_0 - \rho\bullet W_0}{\exp(\sqrt{d_ad_k}t)}
\right) = O
\left( 
\frac{d_k}{\exp(\sqrt{d_ad_k}t)}
\right). 
\end{equation}
Each main iteration requires $O((d_ad_b)^2(d_ad_k)^6)$ flops.
\end{proposition}

\section{Examples and numerical experiments}
\label{sec.examples}
\blue{We performed three main sets of numerical experiments to test the effectiveness of our three main \blue{contributions towards} entanglement detection.  First, we tested the operator $\AC$ and the new PST hierarchy on PICOS implementations.  Second, we tested the EXT and PST \blue{hierarchies} for various levels.  Third, we compared and our FOM and IPM algorithms for PST with its PICOS implementation.  Our experiments use five parametrized collections of entangled states documented in the literature: isotropic (parameter $\lambda \in [0,1]$), Werner (parameter $\lambda \in [0,1]$), $3\times 3$ qutrit (parameter $\alpha \in [0,2.5]$), $3\times 3$ with PPT entanglement (parameter $y \in [0,1]$), $2 \times 4$ with PPT entanglement (parameter $x \in [0,1]$).   Detailed descriptions of these examples are presented in the appendix.   For ease of exposition, in this section we summarize results for the isotropic, qutrit, and $3 \times 3$ collections which are particularly interesting.  The appendix present summaries for the other two collections.
\subsection{Effectiveness of $\AC$ and $\PST$}
The first set of experiments highlights the effectiveness of $\AC$ and $\PST$.  To that end, we compared PICOS implementations of the conic formulations~\eqref{eq.cp.primal} for EXT and~\eqref{eq.conic.pst} for PST, an analogous conic formulation for DPS, and the naive alternative to~\eqref{eq.conic.pst} that does not rely on $\AC$ but instead describes $\PST_k$ via matrices in $\Herm^+(\HC_a \otimes \HC_B)$.  In all cases we used MOSEK SDP solver with default parameters since it is the fastest and most accurate 
general purpose SDP solver available in PICOS.  In particular, MOSEK's primal–dual interior-point method was used with automatic scaling and presolve enabled, double precision arithmetic (float64), and default
feasibility and optimality tolerances of order $10^{-8}$. No solver parameters were manually tuned.  All experiments were run on a single Apple iMac with an Apple M1 CPU (8 cores) and 16 GB RAM, running macOS Tahoe version 26.3.1.}

\begin{figure}[!ht]
\hspace{-.6in}
\begin{tabular}{c} 
\begin{subfigure}{.6\textwidth}
    \centering
        \includegraphics[width=.8\textwidth]{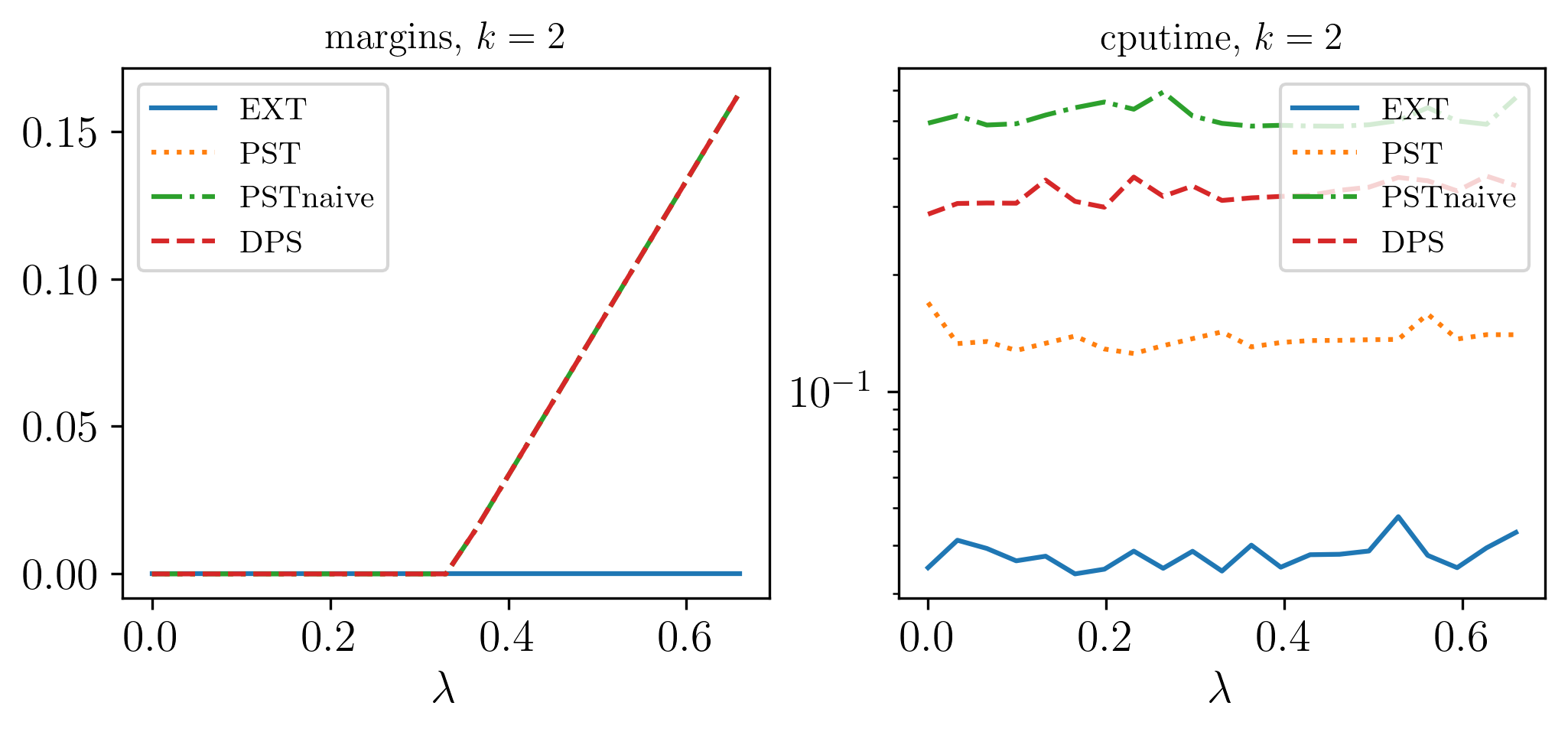}
        \caption{k = 2}
    \end{subfigure} 
    \hspace{-.4in}
    \begin{subfigure}{.6\textwidth}
    \centering
        \includegraphics[width=.8\textwidth]{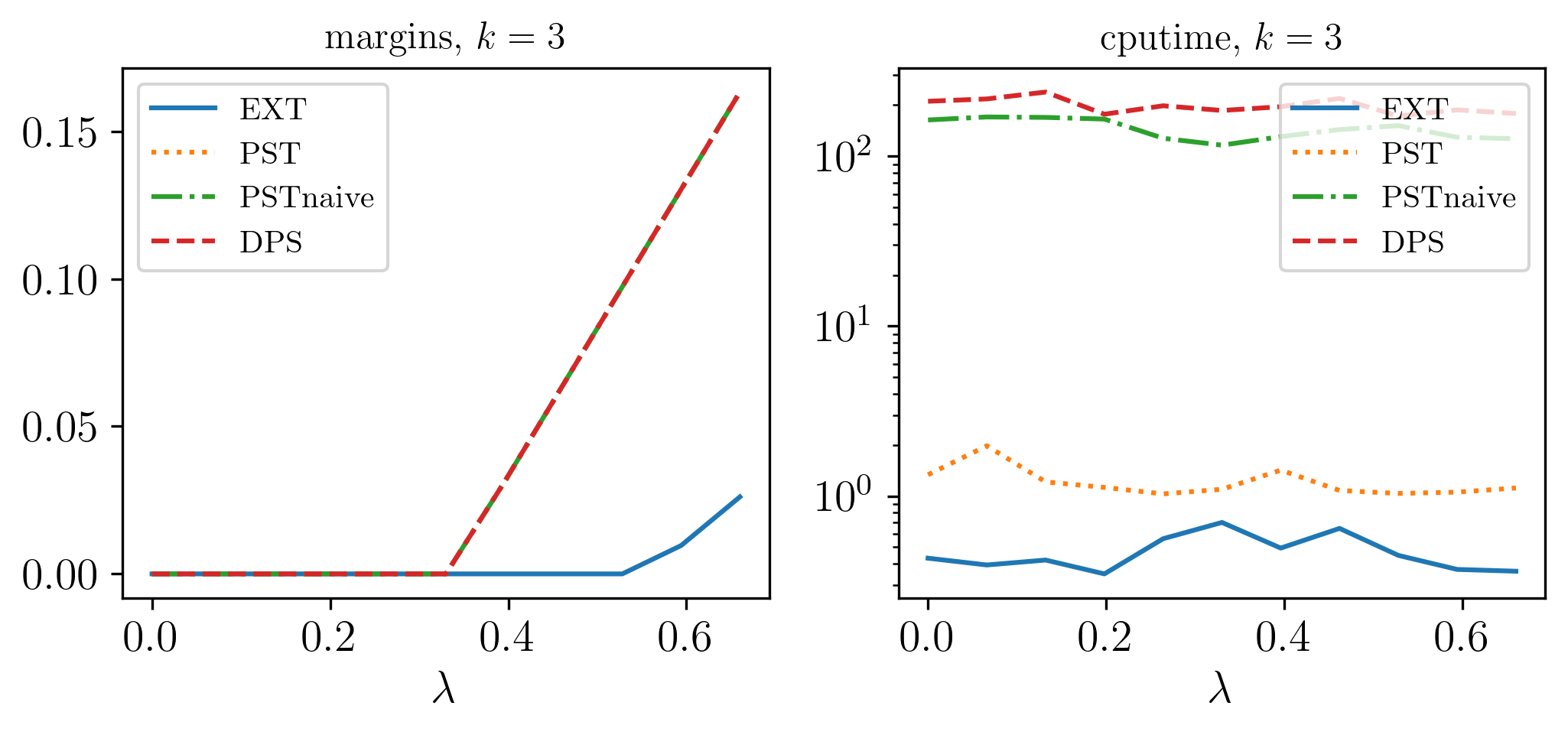}
        \caption{k = 3.}
    \end{subfigure} 
\end{tabular}
\caption{Comparison of EXT, PST, PST (naive), and DPS on the isotropic collection. The lines in the margins plots for PST, PST (naive), and DPS overlap.}\label{figApstiso}
\end{figure}
\begin{figure}[!ht]
 \hspace{-.6in}
\begin{tabular}{c} 
\begin{subfigure}{.6\textwidth}
    \centering
        \includegraphics[width=.8\textwidth]{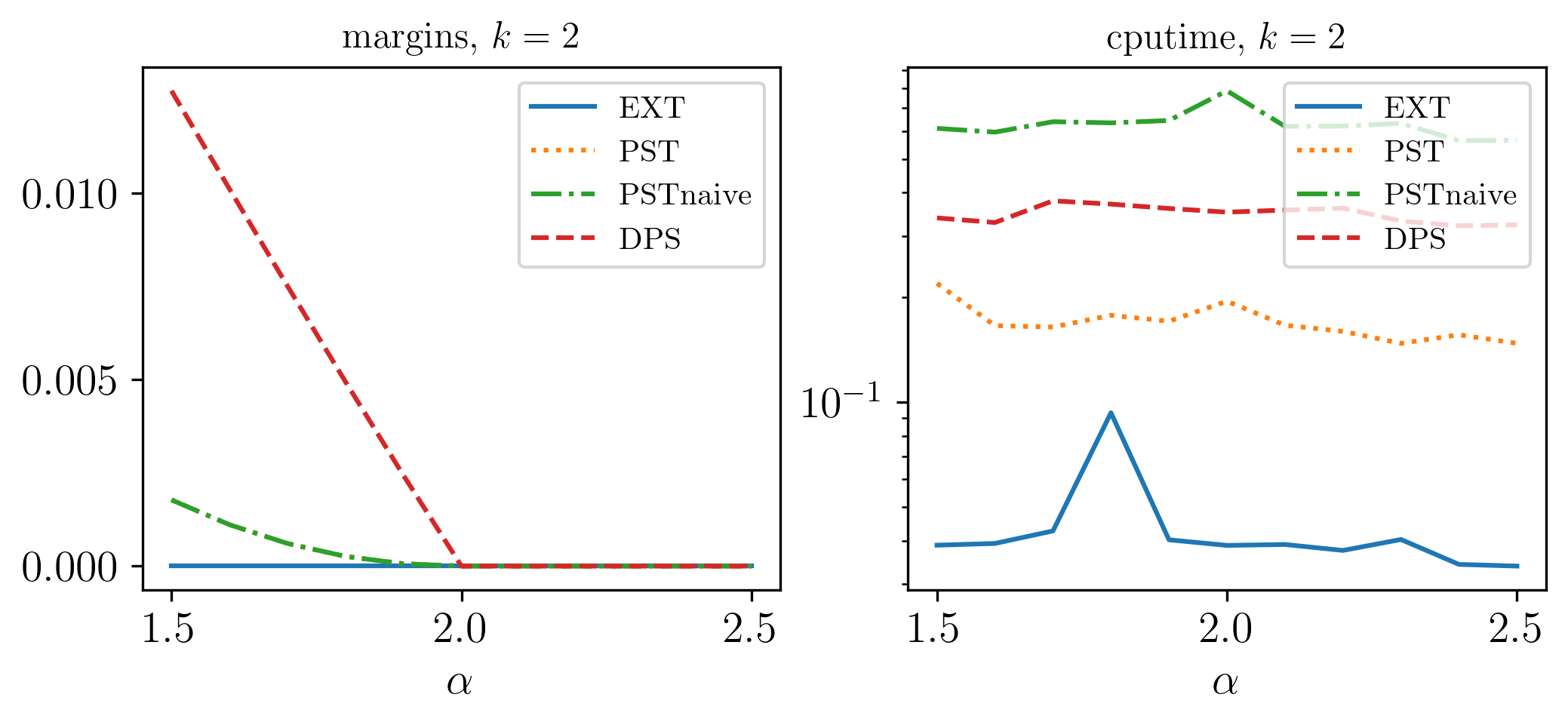}
        \caption{k = 2}
    \end{subfigure} 
    \hspace{-.4in}
    \begin{subfigure}{.6\textwidth}
    \centering
        \includegraphics[width=.8\textwidth]{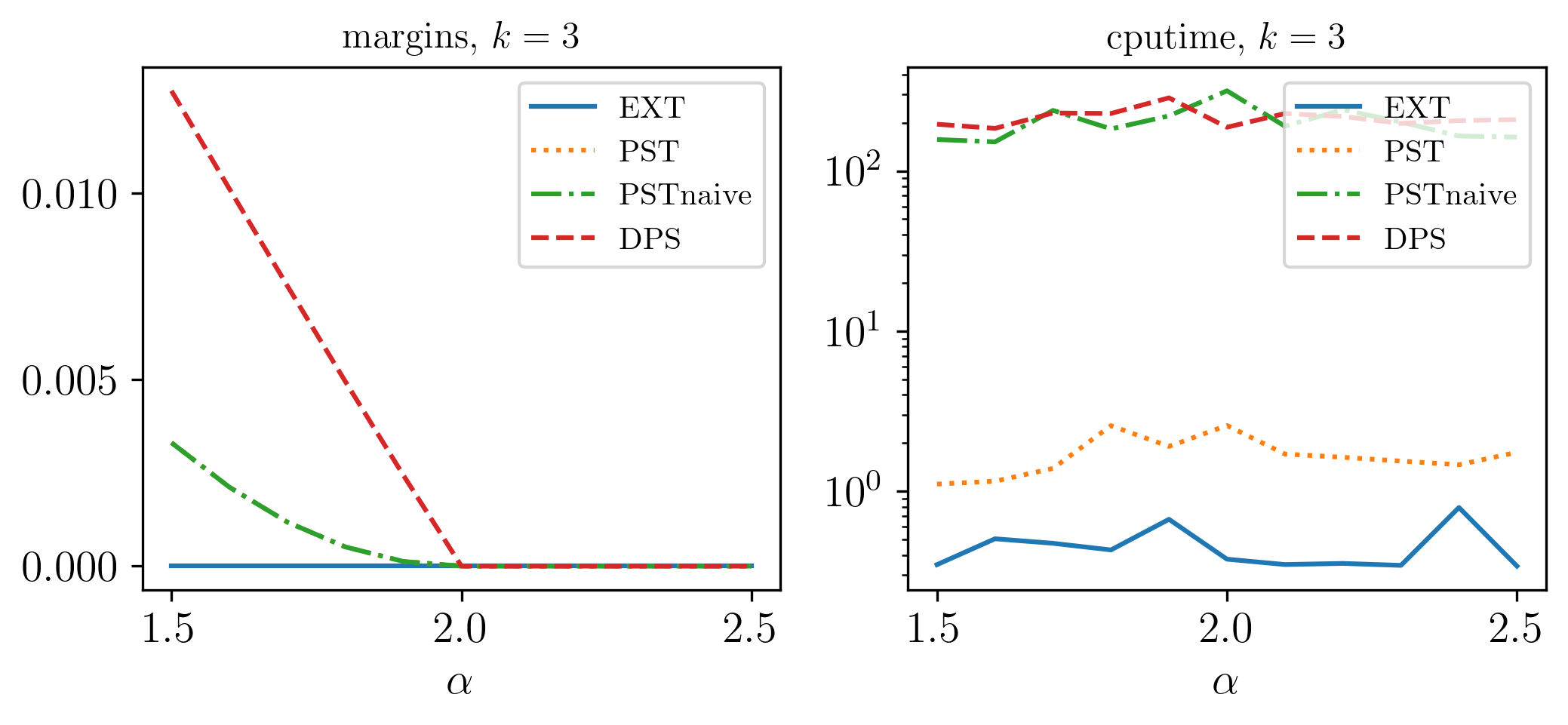}
        \caption{k = 3.}
    \end{subfigure} 
\end{tabular}
\caption{Comparison of EXT, PST, PST (naive), and DPS on the qutrit collection.
The lines in the margins plots for PST and PST (naive) overlap.}\label{figApstqut}
\end{figure}
\begin{figure}[!ht]
 \hspace{-.6in}
\begin{tabular}{c} 
\begin{subfigure}{.6\textwidth}
    \centering
        \includegraphics[width=.8\textwidth]{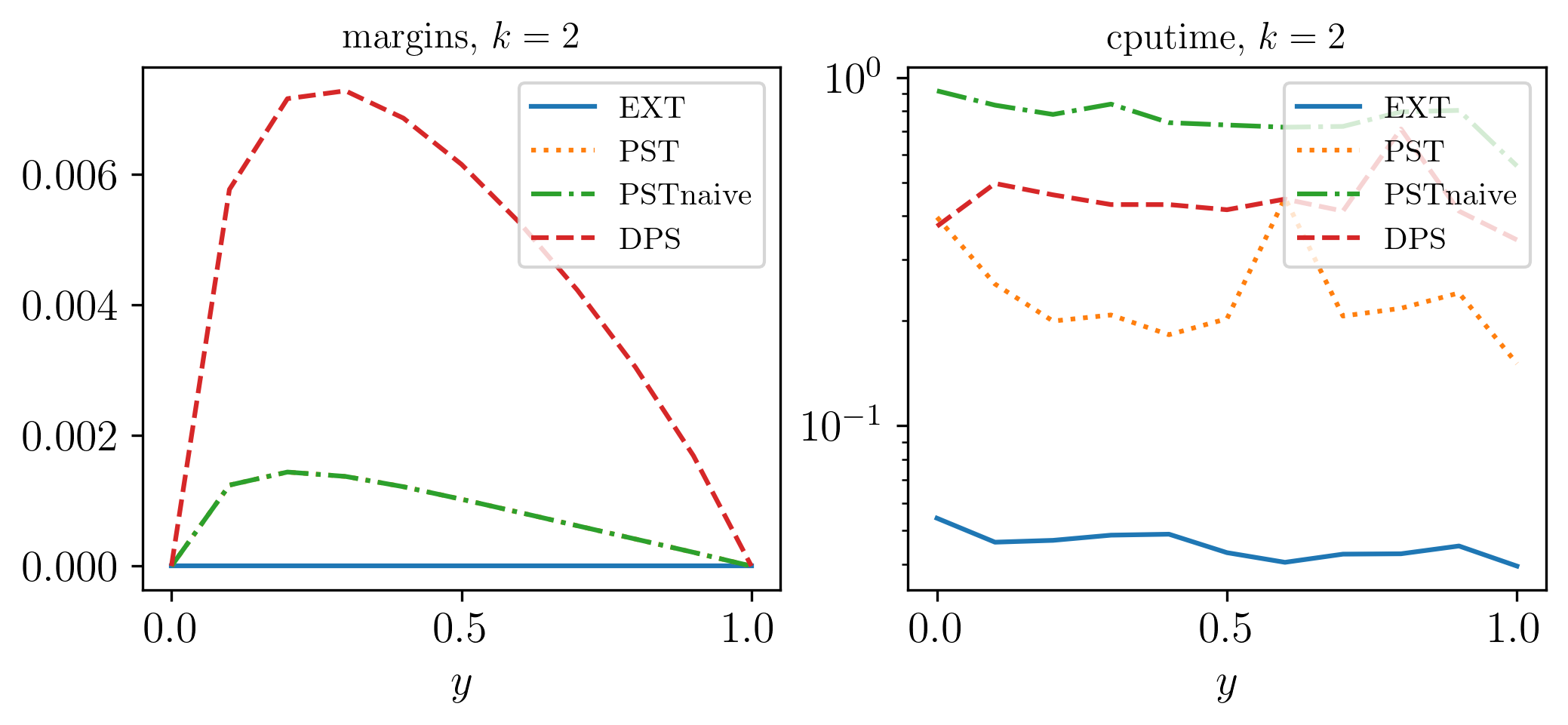}
        \caption{k = 2}
    \end{subfigure} 
    \hspace{-.4in}
    \begin{subfigure}{.6\textwidth}
    \centering
        \includegraphics[width=.8\textwidth]{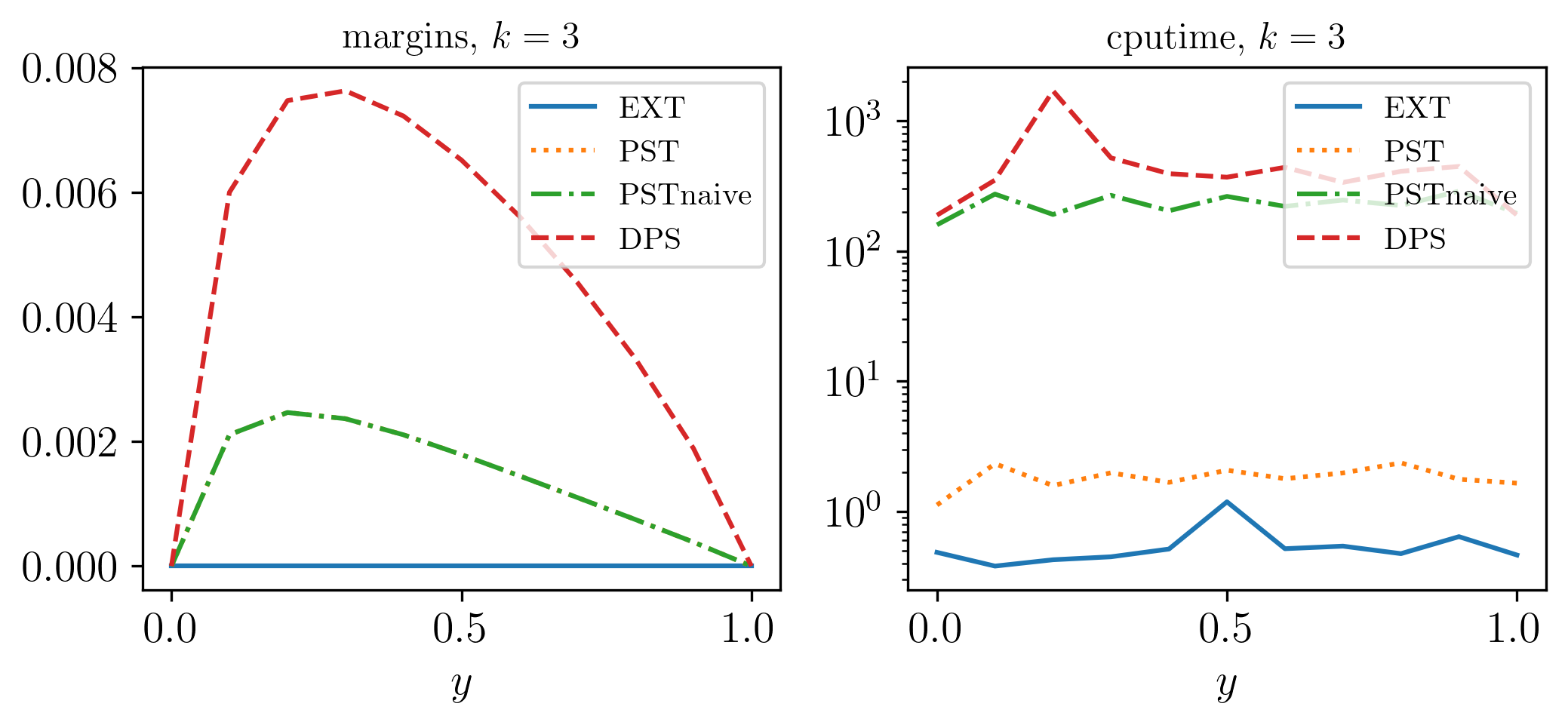}
        \caption{k = 3.}
    \end{subfigure} 
\end{tabular}
\caption{Comparison of EXT, PST, PST (naive), and DPS on the $3 \times 3$ collection.
The lines in the margins plots for PST and PST (naive) overlap.}\label{figApst3by3}
\end{figure}

\blue{Figure~\ref{figApstiso},  Figure~\ref{figApstqut}, and
Figure~\ref{figApst3by3} summarize our results for the isotropic, qutrit, and
$3\times 3$ collections respectively.  
 The following conclusions can  be inferred
from these results:
\begin{itemize}
    \item The combination of the operator $\AC$ and the hierarchy $\PST$ yields a PICOS implementation that runs orders of magnitude faster than both the naive $\PST$ implementation that does not use $\AC$ and \blue{even} the efficient implementation of $\DPS$ that uses $\AC$.
    \item The CPU time in each of the PICOS implementations shows little dependence on the state.
    \item Although DPS is tighter than PST along some directions  $\PST_2$ detects entanglement in all cases when $\DPS_2$ does. 
\end{itemize}
\subsection{$\EXT$ and $\PST$ for various levels}
The second set of experiments compares PICOS implementations of EXT and PST for various levels of $k$. Again these experiments were conducted with MOSEK SDP solver with default parameters and on the same hardware as the first set of experiments.
Figure~\ref{figextpstiso} and Figure~\ref{figextpst3by3} summarize our results for the isotropic, qutrit, and $3\times 3$ collections respectively.
The main conclusion from these results summarized is that at all levels PST is substantially tighter than EXT at the expense of a moderate increase in computational cost.  Hence we concentrate our subsequent discussion to the PST hierarchy only.}
\begin{figure}[!ht]
    \centering
        \includegraphics[width=.5\textwidth]{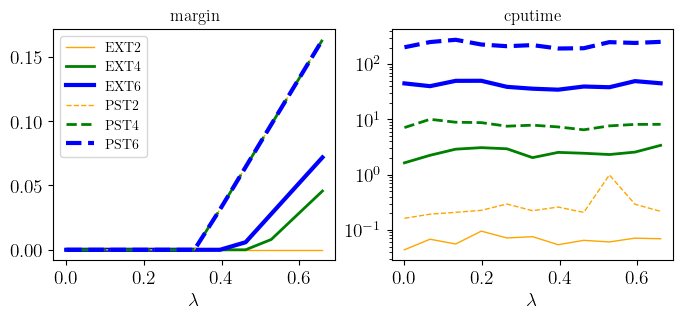}
\caption{Comparison of EXT and PST on the isotropic collection.
The PST lines in the margins plots overlap.
}\label{figextpstiso}
\end{figure}
\vspace{-.2in}
\begin{figure}[!ht]
 \hspace{-.6in}
\begin{tabular}{c} 
\begin{subfigure}{.6\textwidth}
    \centering
        \includegraphics[width=.8\textwidth]{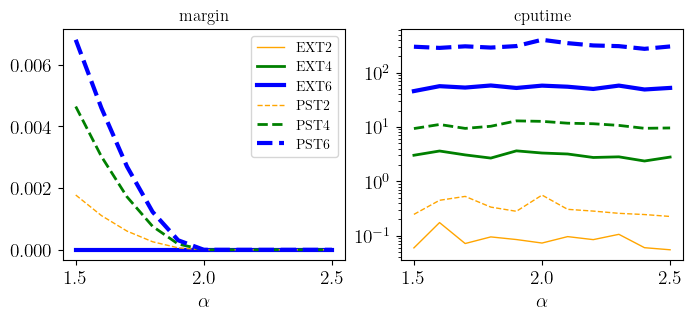}
        \caption{Qutrit states}
    \end{subfigure} 
    \hspace{-.4in}
    \begin{subfigure}{.6\textwidth}
    \centering
        \includegraphics[width=.8\textwidth]{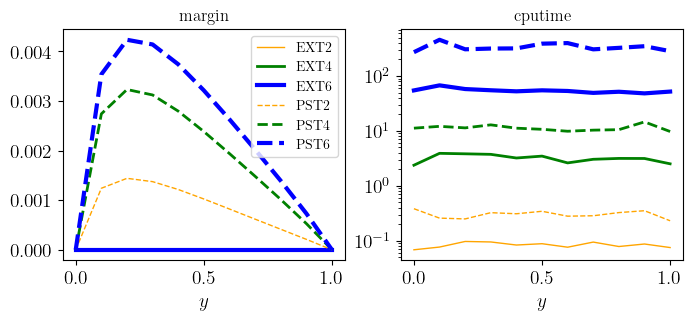}
        \caption{$3\times 3$ states}
    \end{subfigure} 
\end{tabular}
\caption{Comparison of EXT and PST on the qutrit and $3 \times 3$ collections. The EXT lines in the margins plots overlap.}\label{figextpst3by3}
\end{figure}
\blue{
\subsection{Effectiveness of FOM and IPM algorithms for the PST hierarchy}
The third set of experiments highlights the effectiveness of our FOM and IPM algorithms.  To that end, we compared the PICOS implementation of~\eqref{eq.conic.pst} of the PST hierarchy with Python implementations of the FOM and  IPM algorithms described in Section~\ref{sec.PST}. Our implementations of FOM and IPM use plain {\tt NumPy} for numerical linear algebra operations with double precision arithmetic (float64) and no additional stabilization. Thus the implementations of FOM and IPM do not incorporate any advantage over the PICOS implementation. It is important to note that by design, upon termination all three implementations yield a candidate entanglement witness $W\in \PST_k^\circ$ with $I\bullet W = -1$ or a certificate that $\rho$ is equal to or near a point $\tilde \rho\in\PST_k$.  Therefore our comparison of the algorithms is consistent.  We should also note some  important details about the termination criteria for each of the algorithms.  Since MOSEK SDP solver uses a primal-dual infeasible algorithm and does not support the use of warm-starts, PICOS terminates when it solves~\eqref{eq.conic.pst} to optimality and thus it either finds a witness with optimal margin or a certificate that $\rho \in \PST_k$.  By contrast, our FOM and IPM algorithms take advantage of the following more relaxed termination criteria: they terminate when a witness with a safe (but suboptimal) margin is found or when a point $\tilde \rho \in PST_k$ near $\rho$ is found.  For visualization purposes, the latter case is displayed via a (fairly small) negative margin in the figures below.  Figure~\ref{figcompareiso}, Figure~\ref{figcomparequt}, and Figure~\ref{figcompare3by3} summarize our results for the isotropic, qutrit, and $3\times 3$ collections respectively. To illustrate the scalability of FOM, the last two subplots in each case display results for level $3k$ as well, that is, for level up to 18.  We infer the following:
\begin{itemize}
    \item FOM is the fastest entanglement detection procedure (by several orders of magnitude) when the state is sufficiently entangled.  In some instances FOM is faster than the other methods even if we triple the level of the hierarchy.
    \item In the few cases when states are near the boundary of $\PST_k$, FOM may not detect entanglement while the more accurate algorithms IPM or baseline do.  However, in those cases the margin of entanglement is quite small.
\end{itemize}
}

\begin{figure}[!ht]
 \hspace{-.6in}
\begin{tabular}{c} 
\begin{subfigure}{.6\textwidth}
    \centering
        \includegraphics[width=.8\textwidth]{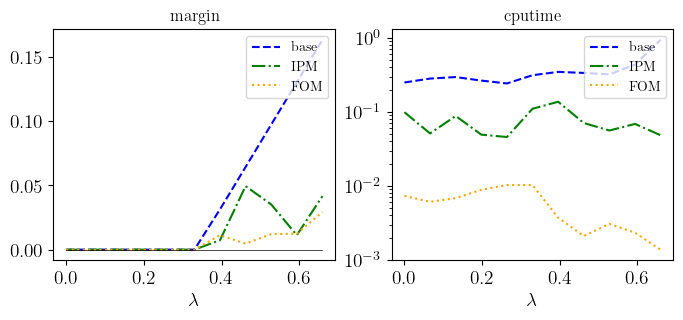}
        \caption{k = 2}
    \end{subfigure} 
    \hspace{-.4in}
    \begin{subfigure}{.6\textwidth}
    \centering
        \includegraphics[width=.8\textwidth]{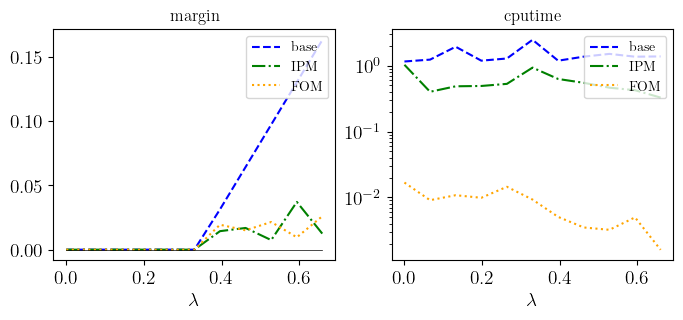}
        \caption{k = 3}
    \end{subfigure} 
\vspace{.2in}
\\
\begin{subfigure}{.6\textwidth}
    \centering
        \includegraphics[width=.8\textwidth]{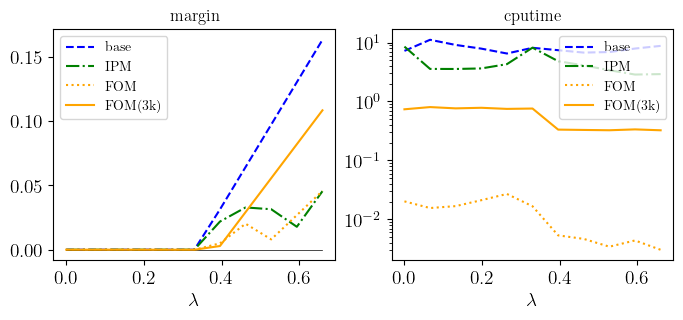}
        \caption{k = 4}
    \end{subfigure} 
    \hspace{-.4in}
    \begin{subfigure}{.6\textwidth}
    \centering
        \includegraphics[width=.8\textwidth]{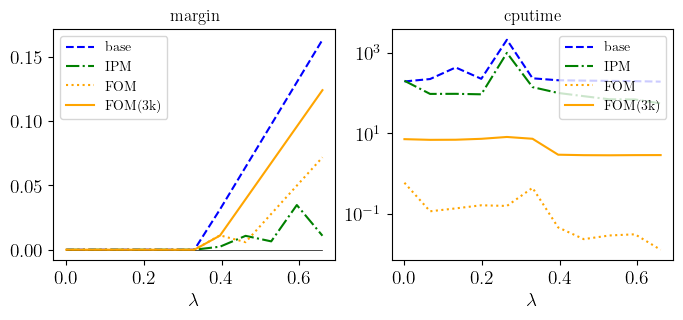}
        \caption{k = 6}
    \end{subfigure}     
\end{tabular}
\caption{Comparison of FOM, IPM, and baseline for PST on the isotropic collection.}\label{figcompareiso}
\end{figure}

\begin{figure}[!ht]
 \hspace{-.6in}
\begin{tabular}{c} 
\begin{subfigure}{.6\textwidth}
    \centering
        \includegraphics[width=.8\textwidth]{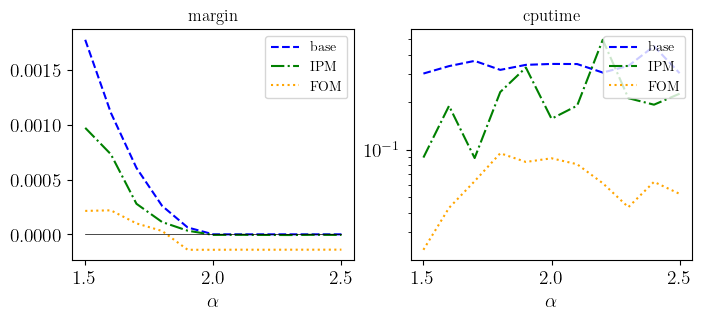}
        \caption{k = 2}
    \end{subfigure} 
    \hspace{-.4in}
    \begin{subfigure}{.6\textwidth}
    \centering
        \includegraphics[width=.8\textwidth]{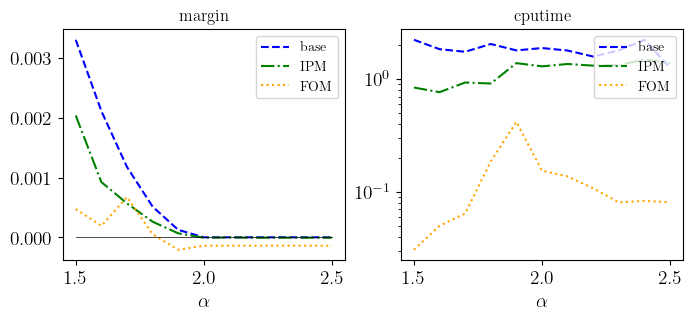}
        \caption{k = 3}
    \end{subfigure} 
\vspace{.2in}
\\
\begin{subfigure}{.6\textwidth}
    \centering
        \includegraphics[width=.8\textwidth]{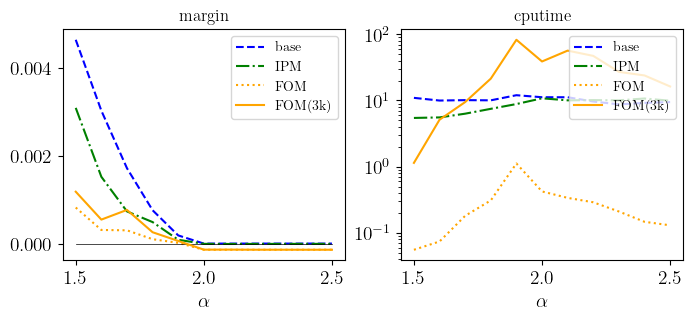}
        \caption{k = 4}
    \end{subfigure} 
    \hspace{-.4in}
    \begin{subfigure}{.6\textwidth}
    \centering
        \includegraphics[width=.8\textwidth]{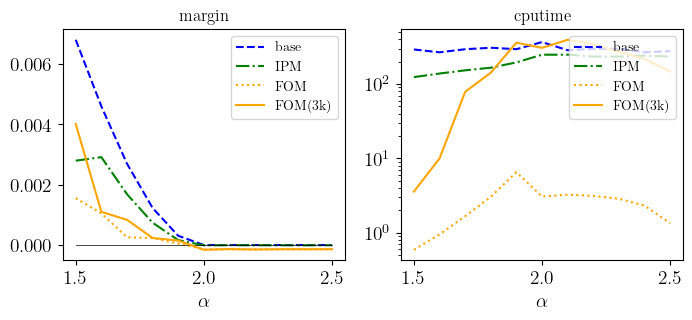}
        \caption{k = 6}
    \end{subfigure}     
\end{tabular}
\caption{Comparison of FOM, IPM, and baseline for PST on the qutrit collection.}\label{figcomparequt}
\end{figure}

\begin{figure}[!ht]
 \hspace{-.6in}
\begin{tabular}{c} 
\begin{subfigure}{.6\textwidth}
    \centering
        \includegraphics[width=.8\textwidth]{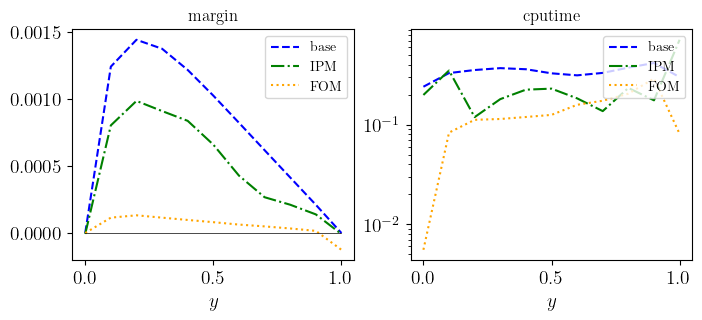}
        \caption{k = 2}
    \end{subfigure} 
    \hspace{-.4in}
    \begin{subfigure}{.6\textwidth}
    \centering
        \includegraphics[width=.8\textwidth]{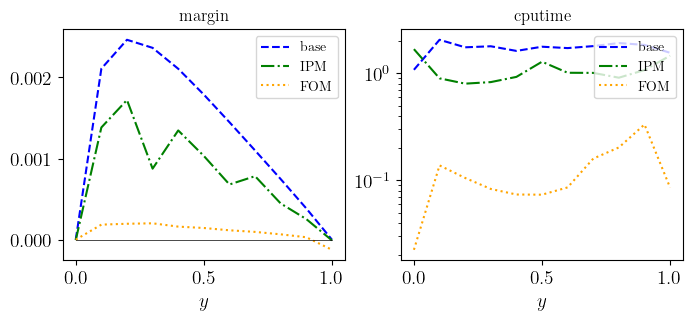}
        \caption{k = 3}
    \end{subfigure} 
\vspace{.2in}
\\
\begin{subfigure}{.6\textwidth}
    \centering
        \includegraphics[width=.8\textwidth]{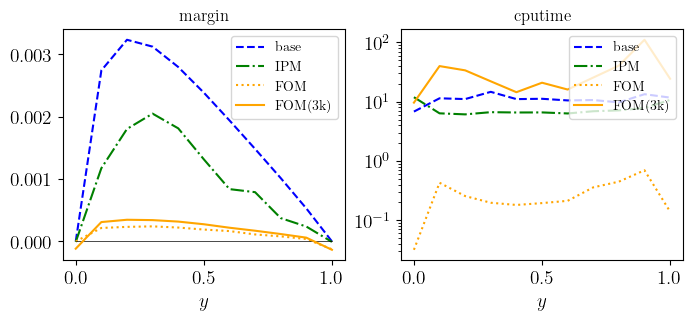}
        \caption{k = 4}
    \end{subfigure} 
    \hspace{-.4in}
    \begin{subfigure}{.6\textwidth}
    \centering
        \includegraphics[width=.8\textwidth]{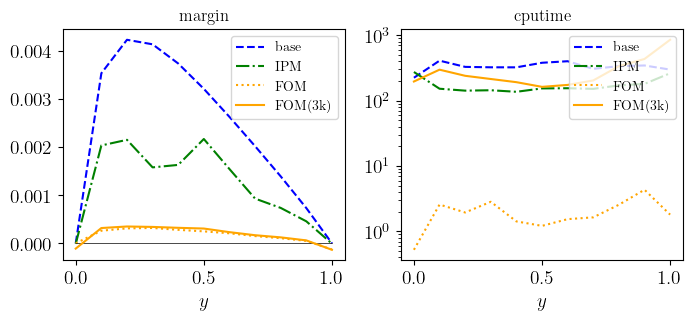}
        \caption{k = 6}
    \end{subfigure}     
\end{tabular}
\caption{Comparison of FOM, IPM, and baseline for PST on  $3 \times 3$ collection.}\label{figcompare3by3}
\end{figure}

\section{Concluding remarks}

We introduced a new semidefinite programming (SDP) hierarchy for entanglement
detection, PST, that combines the computational efficiency of the EXT hierarchy
with a tighter approximation to the set of separable states. We provided
explicit, polynomially scalable formulations of EXT and PST using partition
operators, enabling practical application at higher hierarchy levels without
exponential blow-up.

Two algorithmic frameworks were developed: a least-squares formulation suitable
for first-order methods (FOM) and a conic-programming formulation for a custom
interior-point method (IPM). Both approaches satisfy the Slater condition and
maintain favorable numerical scaling. Our FOM
achieve scalability beyond
previously reported results, while the custom IPM delivers robustness and, in
many cases, an earlier entanglement witness recovery.
Numerical experiments confirm that PST outperforms EXT and that our tailored
algorithms improve both scalability and detection power. These results
demonstrate that careful integration of the problem structure into the
algorithm design can significantly advance entanglement detection.

Further advancements may be possible by exploiting additional structure in the
algorithms discussed here. For instance, the full eigendocomposition used in FOMs can be
circumvented since the algorithm only requires the leading eigenvalues and
eigenvectors. Another interesting pursuit may be to replace
the Hilbert-Schmidt distance minimized in our FOM by a distance that is
truly monotone under quantum channels and thus provides a bona-fide measure for
entanglement~\cite{Ozawa00}. Finally, extending this entanglement detection approach to
multipartite entanglement forms an interesting challenge~(see~\cite{DPS05}
and~\cite{OtfriedGeza09}).

\section*{Acknowledgements}

V.S. is supported by the U.S. Department of Energy, Office of Science, National
Quantum Information Science Research Centers, Co-design Center for Quantum Advantage (C2QA) contract (DE- SC0012704). J.P. is partially supported by the Bajaj Family Chair at the Tepper School of Business, Carnegie Mellon University.


\printbibliography

\newpage

\section{Appendix}
\subsection{Proof of Proposition~\ref{prop.ext}, Proposition~\ref{prop.dps}, and Proposition~\ref{prop.pst}}
Proposition~\ref{prop.ext}, Proposition~\ref{prop.dps}, and Proposition~\ref{prop.pst} are immediate consequences of the following lemma.

\begin{lemma}\label{lemma.symm} Suppose $\Pi:\HC_B \rightarrow \HC_{B}$ is the orthogonal projection onto $\Sym(\HC_{B})$ and
$P:\HC\rightarrow \HC_B$ is a linear operator such that
$P(\HC) = \Sym(\HC_B)$. Then 
\[
\{Y \in L(\HC_B)  \;|\;  Y = \Pi Y \Pi\} = \{P X  P^\dagger \;|\; X\in L(\HC)\}.
\]
In particular,
\begin{align*}
&\{\rho_{aB} \in \Herm(\HC_a\ot\HC_B)  \;|\;  \rho_{aB} = (I_a\ot \Pi) \rho_{aB} (I_a\ot \Pi), \; \rho_{aB}\succeq 0\}\\  &= \{(I_a\ot P) X  (I_a\ot P^\dagger) \;|\; X\in \Herm(\HC_a\ot\HC_B), \, X\succeq 0\}.
\end{align*}
\end{lemma}
\begin{proof}{Proof.}
First we show ``$\supseteq$'': suppose $Y = PX P^\dagger$ for some $X\in L(\HC)$.
Since $\Pi\ket{y} = \ket{y}$ for all $\ket{y}\in \Sym(\HC_B)$ and $P\ket{x}\in \Sym(\HC_B)$ for all $\ket{x}\in \HC$, it immediately follows that $\Pi P = P$ and thus $\Pi Y \Pi=  \Pi PX P^\dagger\Pi = PX P^\dagger= Y$.

Next we show ``$\subseteq$'': suppose $Y = \Pi Y \Pi$.  Then $\Pi Y = Y$ and so $Y = PZ$ for some $Z\in L(\HC_B,\HC)$ because each of the the columns of $Y=\Pi Y$ belongs to $\Sym(\HC_B)$.  Again since $Y = \Pi Y \Pi$, it follows that $Y=Y\Pi =PZ\Pi$.  Since each of the columns of $\Pi Z^\dagger$ belongs to $\Sym(\HC_B)$ we also have $\Pi Z^\dagger = PX^\dagger \Leftrightarrow Z \Pi = X P^\dagger$ for some $X\in L(\HC)$.  Therefore $Y = PZ\Pi = PXP^\dagger$.
\end{proof}

\subsection{Additional details on the construction of $\AC$ and $\AC^\dagger$ and Proof of Proposition~\ref{prop.norm}}

Let $[d]^k$ denote the $d^k$ sequences of length $k$ from $\{0,1,\dots,d-1\}$ and $[d]^{k\uparrow} \subseteq [d]^k$ denote the $d_k$ non-decreasing sequences of length $k$ from $\{0,1,\dots,d-1\}$.  We will write $\vec{i}$ to denote a generic sequence $i_1i_2\cdots i_{k}$ in $[d]^k$.

We will rely on the following canonical bases and coordinate systems for the spaces $\HC_b,\, \HC_B,$ and $\HC$ for a conveniently chosen $\HC$.
\begin{itemize}
\item For $\HC_b$: let $\{\ket{0},\dots,\ket{d-1}\}$ be an orthonormal basis for $\HC_b$.  Accordingly, the coordinates for $\HC_b$ will be $0,1,\dots,d-1$.
\item  For $\HC_B$: take the basis $\{\vec{\ket{i}} \;|\; \vec{i}\in [d]^k\}$.  
Accordingly, the set of coordinates for $\HC_B$ will be $\{\vec{i} \;|\; \vec{i}\in [d]^k\}$ ordered lexicographically.  For example, for $d=2$ and $k=3$ the coordinates are
\[
000, 001, 010, 011, 100, 101, 110, 111;
\]
and for $d=3$ and $k=2$ the coordinates are
\[
00, 01, 02, 10, 11, 12, 21, 20, 22.
\]
\item For $\HC$: let $\HC\subseteq \HC_B$ be the subspace spanned by the basis
$\{\vec{\ket{i}} \;|\; \vec{i}\in [d]^{k\uparrow}\}$.  Accordingly, the set of coordinates 
for $\HC$ will be $\{\vec{i} \;|\; \vec{i}\in [d]^{k\uparrow}\}$ ordered lexicographically.  For example, for $d=2$ and $k=3$ the coordinates for $\HC$ are
\[
000, 001, 011, 111;
\]
and for $d=3$ and $k=2$ the coordinates are
\[
00, 01, 02, 11, 12, 22.
\]

\end{itemize}


Let $P:\HC\rightarrow \HC_B$ denote the {\em scaled partition mapping} that maps the $\vec{i}$-th unitary vector from $\HC$ to the {\em normalized} indicator vector of the set of permutations of the sequence $\vec{i}$. Since the columns of $P$ are normalized, it follows that $P^\dagger P = I_{\HC}$.
For example, the mapping $P$ has the following matrix representation in the above  coordinate systems for $d=2,k=3$ and for $d=3,k=2$ respectively:
\[
P = \begin{pmatrix} 
1&0&0&0\\
0&1/\sqrt{3}&0&0\\
0&1/\sqrt{3}&0&0\\
0&0&1/\sqrt{3}&0\\
0&1/\sqrt{3}&0&0\\
0&0&1/\sqrt{3}&0\\
0&0&1/\sqrt{3}&0\\
0&0&0&1\\
\end{pmatrix},
\;\;
P = \begin{pmatrix} 
1&0&0&0&0&0\\
0&1/\sqrt{2}&0&0&0&0\\
0&0&1/\sqrt{2}&0&0&0\\
0&1/\sqrt{2}&0&0&0&0\\
0&0&0&1&0&0\\
0&0&0&0&1/\sqrt{2}&0\\
0&0&1/\sqrt{2}&0&0&0\\
0&0&0&0&1/\sqrt{2}&0\\
0&0&0&0&0&1\\
\end{pmatrix}.
\]
The crux of the operators $\AC$ and $\AC^\dagger$ is the following operator $M: L(\HC_b) \rightarrow L(\HC)$
\[
M(U):= P^\dagger \left(U \ot I_{b_{2:k}}\right) P.
\]
Indeed, a simple calculation shows that $\AC^\dagger(W) = (\mathbb{I}_a\ot M)(W)$ where $\mathbb{I}_a$ denotes the identity operator $\mathbb{I}_a: L(\HC_a)\rightarrow L(\HC_a)$.
It is evident that for $U\in L(\HC_b)$ the matrix representation of $M(U) \in L(\HC)$ can be written as follows in terms of the coordinates of $U$:
\[
M(U) = \sum_{i,j=0}^{d-1}  M_{ij} U_{ij}
\]
where  $ M_{ij} =  M(E_{ij})$ for each basis operator $E_{ij} = \ket{i}\bra{j} \in L(\HC_b)$.  

We next give an explicit expression for each $M_{ij}$ in the coordinates of $\HC$. To that end, we will rely on the following notational convention.   A sequence $\vec{\ell} \in [d]^{k\uparrow}$ can be coded as a tuple $(\ell_0,\cdots,\ell_{d-1})$ where $\ell_i$ denotes the number of appearances of $i$ in $\vec{\ell}$.  In other words, we can identify $\vec{\ell} \in [d]^{k\uparrow}$ with a tuple of  of nonnegative integers $(\ell_0,\cdots,\ell_{d-1})$ such that $\sum_{i=0}^{d-1} \ell_i = k$. 
We will also use the following convention: given $\vec{\ell}\in[d]^{(k-1)\uparrow}$ and $i\in[d]$ we write $\vec{\ell}+i\in[d]^{k\uparrow}$ to denote the sequence obtained by adding $i$ to the sequence $\vec{\ell}$.  Thus, in tuple notation $\vec{\ell}+i = (\ell_0,\dots,\ell_i+1,\dots,\ell_d)$. By relying on these pieces of notation and some simple calculations it follows that the nonzero entries of $M_{ij}$ are
\begin{equation}\label{eq.Mentries}
(M_{ij})_{(\vec{\ell}+i,\vec{\ell}+j)} = \frac{\sqrt{(\ell_i+1)(\ell_j+1)}}{k} \text{ for } \vec{\ell}\in [d]^{(k-1)\uparrow}.
\end{equation}
For example, for $d=2$ and $k=3$ we have 
\[
M(U) = M_{00} U_{00}+  M_{01} U_{01} +  M_{10} U_{10} +  M_{11} U_{11}
\]
where 
\[
 M_{00} = \begin{pmatrix} 
1&0&0&0\\
0&2/3&0&0\\
0&0&1/3&0\\
0&0&0&0
\end{pmatrix}, \;  M_{01} = \begin{pmatrix} 
0&1/\sqrt{3}&0&0\\
0&0&2/3&0\\
0&0&0&1/\sqrt{3}\\
0&0&0&0
\end{pmatrix},
\;  M_{11} = \begin{pmatrix} 
0&0&0&0\\
0&1/3&0&0\\
0&0&2/3&0\\
0&0&0&1
\end{pmatrix}
\]
and $M_{10} = M_{01}^\dagger.$

The matrices $M_{ij}$ are quite sparse and have the same sparsity pattern.  Indeed, from~\eqref{eq.Mentries} it readily follows that the number of nonzeros in each $d_k \times d_k$ matrix $M_{ij}$ is $d_{k-1} < d_k \ll d_k^2$.  Figure~\ref{fig.sparseM} illustrates the sparsity pattern of $M^\dagger$ for $d_a=d_b=3$ and $k=3$ when reshaped as matrix.  In this example the number of nonzeros in each row of $M^\dagger$  is $d_{k-1} = {3+1 \choose 2} = 6$.

\begin{figure}[!t]
        \includegraphics[width=5.75in]{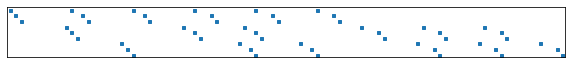}
\caption{Sparsity pattern of $M^\dagger$ reshaped as a $d^2 \times d_k^2$ matrix for $d=3$ and $k= 3$}\label{fig.sparseM}
\end{figure}

\begin{proof}{Proof of Proposition~\ref{prop.norm}.}
Since $\AC^\dagger =  \mathbb{I}_a \ot M$, it follows that 
$\AC \AC^\dagger = I_{d_a^2} \ot (M^\dagger M)$ if $\AC$ and $M$ are reshaped as matrices.  Therefore $ \|\AC^\dagger \AC\|=\|\AC \AC^\dagger\|= \|M^\dagger M\|$ and $\|(\AC \AC^\dagger)^{-1}\|= \|(M^\dagger M)^{-1}\|.$ 
We next describe the entries of $M^\dagger M$.  To that end, observe that the $(ij, i'j')$ entry of $M^\dagger M$ is 
\[
(M^\dagger M)_{ij, i'j'} = M_{ij} \bullet M_{i'j'}.
\]
The construction of $M_{ij}$ implies that the only nonzero entries of $M^\dagger M$ are 
the $(ij,ij)$ and $(ii,jj)$ entries for $i,j \in [d]$. Some simple  calculations show that for $i\in [d]$
\begin{equation}\label{eq.ak}
M_{ii}\bullet M_{ii} = \sum_{\vec{\ell} \in [d]^{(k-1)\uparrow}} \frac{(\ell_i+1)^2}{k^2} = \frac{1}{k^2} \sum_{\ell=1}^k \ell^2 \cdot (d-1)_{k-\ell},
\end{equation}
and for $i,j\in[d]$ with $i\ne j$ both
\begin{equation}\label{eq.bk}
M_{ij}\bullet M_{ij} = \sum_{\vec{\ell} \in [d]^{(k-1)\uparrow}} \frac{(\ell_i+1)(\ell_j+1)}{k^2} = \frac{1}{k^2} \sum_{\ell_0=1}^k\sum_{\ell_1 = 1}^{k+1-\ell_0} \ell_0\ell_1 \cdot (d-2)_{k+1-\ell_0-\ell_1},
\end{equation}
and 
\begin{equation}\label{eq.ck}
M_{ii}\bullet M_{jj} = \displaystyle\sum_{\vec{\ell} \in [d]^{(k-2)\uparrow}} \frac{(\ell_i+1)(\ell_j+1)}{k^2}=\frac{1}{k^2} \sum_{\ell_0=1}^{k-1}\sum_{\ell_1 = 1}^{k-\ell_0} \ell_0\ell_1 \cdot (d-2)_{k-\ell_0-\ell_1}. 
\end{equation}
To ease notation, let $a_k, b_k ,c_k$ denote the right-most expressions in~\eqref{eq.ak},~\eqref{eq.bk},and~\eqref{eq.ck} respectively.  
A combinatorial argument shows that  $a_k = b_k + c_{k}$. Therefore 
$M^\dagger M$ reshaped as a matrix can be written as
\[
b_k I_{d^2} + c_{k}\ket{v}\bra{v}
\]
where $\ket{v}$ is a $d^2$ vector with $d$ entries equal to one and all other entries equal to zero.  Hence the smallest and largest eigenvalues of both $M^\dagger M$ and $\AC \AC^\dagger$ are respectively $b_k$ and $b_k+dc_{k}$. 

A simple calculation shows that $I_d$ reshaped as a $d^2$ vector is an eigenvector of $M^\dagger M$ with  eigenvalue $b_k+dc_k$, that is $M(I_d) = (b_k+dc_{k})I_d$.  Since $M(I_{d}) = I_{\HC}$, it follows that
\[
(b_k + dc_{k})d = (b_k + dc_{k})\|I_d\|^2 = \|M(I_d)\|^2 = \|I_{\HC}\|^2 = d_k.
\]
Therefore 
\[
\|M^\dagger M\| = b_k + dc_k = \frac{d_k}{d}.
\]
Since $c_k \le b_k$, it follows that $(d+1)b_k \ge b_k + dc_k = d_k/d.$  Hence $b_k \ge d_k/d(d+1)$ and consequently
\[
\|(M^\dagger M)^{-1}\| = \frac{1}{b_k} \le \frac{d(d+1)}{d_k}.
\]
\end{proof}

\subsection{Proofs of Proposition~\ref{prop.fom} and Proposition~\ref{prop.fom.pst} }

The main update in Algorithm~\ref{algo.fom}, namely, $$(u_{t},X_{t+1}):=\FOMupdate(t),$$ can be performed via one any of the following three algorithmic schemes: Frank-Wolfe (FW),  Projected gradient (PG), and  fast projected gradient (FPG) as we next detail.

\begin{enumerate}
\item  {\bf Frank-Wolfe update~\cite{Braun2022,Jaggi2013}}:
\[
\begin{aligned}    
u_{t}&:=\AC(X_{t})-\rho, \quad S_t:=\text{argmin}_{S\in \D(\HC_a\ot \HC)} \AC^\dagger(u_t)\bullet S\\
X_{t+1} &:= X_t + \gamma_t (S_t-X_t) \text{ for some }\gamma_t \in [0,1].
\end{aligned}
\]
The second step above relies on a linear oracle for $\D(\HC_a\ot \HC)$.
We choose use the stepsize $\gamma_t$ via the following ``line-search'' procedure:
$$
\gamma_t = \text{argmin}_{\gamma \in [0,1]} \|\AC(X_t+\gamma(S_t-X_t))-\rho\|^2 = 
-\frac{(\AC(X_t) - \rho)\bullet(\AC(S_t-X_t))}{\|\AC(S_t-X_t)\|^2}.
$$

\item {\bf Projected gradient update~\cite{BenTal2021,GutmPena2023}:}
\[
\begin{aligned}  
X_{t+1} &:= \Pi_{\D(\HC_a\ot \HC)}(X_t - \tau_t \AC^{\dagger}(\AC(X_t)-\rho)) \text{ for some stepsize } \tau_t>0 \\
u_{t} &:=\frac{\sum_{i=0}^{t-1}\tau_i(\AC(X_i)-\rho)}{\sum_{i=0}^{t-1}\tau_i}
\end{aligned}
\]
The first step above relies on the orthogonal projection mapping $\Pi_{\D(\HC_a\ot \HC)}:\Herm(\HC_a\otimes \HC)\rightarrow \D(\HC_a\ot \HC)$ defined as
\[
\Pi_{\D(\HC_a\ot \HC)}(X)= \min_{Y\in \D(\HC_a\ot \HC)}\|Y-X\|.
\]
We choose the stepsize $\tau_t$ via backtracking so that the following sufficient condition for convergence holds: pick the largest $\tau_t > 0$ such that
\begin{equation}\label{eq.stepsize}
f(X_{t+1}) \le f(X_t) + \nabla f(X_t)\bullet (X_{t+1}-X_t) + \frac{1}{2\tau_t}\|X_{t+1} - X_t\|^2
\end{equation}
for the objective function $f(X):= \frac{1}{2}\|\AC(X)- \rho\|^2$.

\item {\bf Fast projected gradient update~\cite{BenTal2021,GutmPena2023}:}.  This update relies on some extra iterates $\tilde X_t, S_t$ and on the following convention: for a sequence of positive stepsizes $\tau_t, \; t=0,1,\dots$ we let $\theta_t, \; t=0,1,\dots$ be the sequence
\[
\theta_t := \frac{\tau_t}{\sum_{i=0}^t \tau_i}.
\]
The update $(u_{t},X_{t+1}) = \FOMupdate(t)$ via fast projected gradient is as follows
\[
\begin{aligned}
\tilde X_t&:=(1-\theta_t)X_t + \theta_t S_{t-1}\\
S_{t} &:= \Pi_{\D(\HC_a\ot \HC)}(S_{t-1} - \tau_t \AC^{\dagger}(\AC(\tilde X_t)-\rho)) \text{ for some } \tau_t > 0 \\
X_{t+1} &:= (1-\theta_t)X_t + \theta_t S_{t} \\
u_t&:=\frac{\sum_{i=0}^{t-1}\tau_i(\AC(Y_i)-\rho)}{\sum_{i=0}^{t-1}\tau_i}
\end{aligned}
\]    
When $t=0$ we set $S_{-1} := X_0$ in the first step above.

Again we choose the stepsize $\tau_t$ via backtracking so that the following sufficient condition for convergence holds: pick the largest $\tau_t > 0$ such that
\begin{equation}\label{eq.stepsize.fpg}
\frac{\tau_t}{\theta_t} \left(f(X_{t+1}) - f(\tilde X_t) - \nabla f(\tilde X_t)\bullet (X_{t+1}-\tilde X_t)\right) \le \frac{1}{2}\|S_{t} - S_{t-1}\|^2
\end{equation}
for the objective function $f(X):= \frac{1}{2}\|\AC(X)- \rho\|^2$.

\end{enumerate}

\begin{proof}{Proof of Proposition~\ref{prop.fom}.} We prove~\eqref{eq.fom.conv} for each of the three updates.

First, for the FW update, we rely on the following key property~\cite{Braun2022,Jaggi2013}:
 for $u_t:=\AC(X_t)-\rho$ and for our choice of stepsize the iterates of Algorithm~\ref{algo.fom} satisfy
\[
\gap(X_t,u_t) \le \frac{C}{t+2}
\]
where $C$ is the {\em curvature} of the objective function $X\mapsto f(X):= \frac{1}{2}\|\AC(X)- \rho\|^2$ on $\Delta$, namely:
\begin{equation}\label{curv}
C = \max_{X,S\in \D(\HC_a\ot\HC), \gamma \in (0,1)}\frac{f(X+\gamma(S-X)) - \gamma\nabla f(X)\bullet (S-X) - f(X)}{\gamma^2/2}.
\end{equation}
To get~\eqref{eq.fom.conv}, it suffices to show that $C\le 4$.  Indeed, 
from~\eqref{curv}, the inequality between the Frobenius and nuclear norms, and Proposition~\ref{prop.prop.A} it readily follows that
\begin{multline*}
C= \max_{X,S\in \D(\HC_a\ot\HC)}\|\AC(X-S)\|^2 
\le \max_{X,S\in \D(\HC_a\ot\HC)}(\|\AC(X)\|_* + \|\AC(S)\|_*)^2
\\
= \max_{X,S\in \D(\HC_a\ot\HC)}(\|X\|_*+\|S\|_*)^2
 = 4. 
\end{multline*}
Second, for the PG update, 
will rely on the following property of the projected gradient algorithm established in~\cite{GutmPena2023}:
for $u_t$ as chosen in Algorithm~\ref{algo.fom} and stepsizes that satisfy~\eqref{eq.stepsize} it holds that
\begin{equation}\label{eq.gap.bound}
\gap(X_t,u_t)  \le \frac{\max_{X\in \D(\HC_a\ot\HC)} \|X-X_0\|^2/2}{\sum_{i=0}^{t-1}\tau_i}.
\end{equation}
To get~\eqref{eq.fom.conv}, we next bound each of the terms $\max_{X\in \Delta}\|X-X_0\|$ and the $\tau_i$.  The inequality between the Frobenius and nuclear norms implies that $\|X-Y\| \le \|X-Y\|_* \le \|X\|_* + \|Y\|_* = 2$ for all $X,Y\in \D(\HC_a\ot\HC)$.  Therefore,
$
\max_{X\in \D(\HC_a\ot\HC)} \|X-X_0\|^2/2 \le 2.
$
On the other hand, some algebraic calculations and Proposition~\ref{prop.norm} show that for all $X,Y \in \D(\HC_a\ot\HC)$
\[
f(Y)-f(X)-\nabla f(X)\bullet (Y-X) = \frac{1}{2}\|\AC(X-Y)\|^2 \le \frac{1}{2} \|\AC^\dagger \AC\| \|X-Y\|^2 \le \frac{d_k}{2d_b}  \|X-Y\|^2 .
\]
Thus each $\tau_i \ge d_b/d_{k}$.
Plugging this bound in~\eqref{eq.gap.bound} we get
\[
\gap(X_t,u_t)\le\frac{2}{dt/d_k} = \frac{2d_k}{d_bt}.
\]
Third, for FPG, we rely on the following property of the fast projected gradient algorithm established in~\cite{GutmPena2023}:
for $u_t$ as chosen in Algorithm~\ref{algo.fom} and stepsizes that satisfy~\eqref{eq.stepsize.fpg} it holds that
\begin{equation}\label{eq.gap.bound.fpg}
\gap(X_t,u_t)  \le \frac{\max_{X\in  \D(\HC_a\ot\HC)} \|X-X_0\|^2/2}{\sum_{i=0}^{t-1}\tau_i}.
\end{equation}
From the above PG case, we get $
\max_{X\in  \D(\HC_a\ot\HC)} \|X-X_0\|^2/2 \le 2$.  To get~\eqref{eq.fom.conv}, we next bound $1/(\sum_{i=0}^{t-1}\tau_i)$.  Some algebraic calculations 
and Proposition~\ref{prop.norm} 
show that for $X_{t+1},\tilde X_t \in  \D(\HC_a\ot\HC)$
\begin{multline*}
f(X_{t+1})-f(\tilde X_t)-\nabla f(\tilde X_t)\bullet (X_{t+1}-\tilde X_t) = \frac{1}{2} \|\AC(X_{t+1}-\tilde X_t)\|^2 = \frac{\theta_t^2}{2} \|\AC(S_t-S_{t-1})\|^2\\
\le \frac{\theta_t^2}{2}\|\AC^\dagger \AC\| \|S_t-S_{t-1}\|^2
\le \frac{d_k\theta_t^2}{2d_b} \|S_t-S_{t-1}\|^2.
\end{multline*}
Thus~\eqref{eq.stepsize.fpg} holds provided that $\tau_t \theta_t$ does not exceed  $d/d_k$.  Hence~\cite[Lemma 1]{GutmPena2023} implies that the stepsizes $\tau_t$ satisfy
\[
\frac{1}{\sum_{i=0}^{t-1}\tau_i}\le \frac{4d_k}{d_b(t+1)^2}.
\]
Plugging this bound and the bound $
\max_{X\in \Delta} \|X-X_0\|^2/2 \le 2$ in~\eqref{eq.gap.bound.fpg} we get
\[
\gap(X_t,u_t)\le \frac{8d_k}{d_b(t+1)^2}.
\]
\end{proof}

\begin{proof}{Proof of Proposition~\ref{prop.fom.pst}.} The proof is nearly identical to that of   Proposition~\ref{prop.fom}.  The only difference is that some constants are a bit larger because of the additional terms in the objective function and constraints.
\end{proof}

\subsection{Description of the examples in Section~\ref{sec.examples}}    
\noindent
{\bf Isotropic states.}
Let $d_a=d_b=d$ and consider the family of isotropic states
\begin{equation}\label{eq:isotropic}
    \rho = \lambda \dya{\psi} + \frac{1-\lambda}{d^2-1}(I_{ab} -
    \dya{\psi}),
\end{equation}
where $\lambda \in [0,1]$ and
$
\ket{\psi} = \frac{1}{\sqrt{d}}\sum_{i=1}^d \ket{i}_a\otimes\ket{i}_b.
$
%
The state $\rho$ is entangled if $\lambda \in (1/d,1]$ and separable if
$\lambda \in [0, 1/d]$ as shown in~\cite{Horodecki98}. 

\bigskip 
 
\noindent 
{\bf Werner states.}
Let $d_a=d_b=d$ and consider the family of Werner states~\cite{Werner89},
%
\begin{equation}
    \rho = \frac{\lambda}{d(d+1)} (I_{ab} + W) + \frac{1-\lambda}{d(d-1)}(I_{ab} - W)
    \label{eq:wrnr}
\end{equation}
where $\lambda \in [0,1]$ and
\begin{equation}
    W = \sum_{ij} \ket{i}\bra{j}\otimes\ket{j}\bra{i}.
    \label{eq:swpd}
\end{equation}
%
This state is entangled for
$\lambda \in [0,1/2)$ and separable for $\lambda\in [1/2,1]$. 

\bigskip 

\noindent
{\bf Qutrit states.}
Let $d_a=d_b=3$ and consider the two qutrit state described and studied
in~\cite{Horodecki01},
\begin{equation}
    \rho = \frac{2}{7}\ket{\psi_{+}}\bra{\psi_{+}} + \frac{\alpha}{7}\sigma_{+} +
\frac{5 -  \alpha}{7}S \sigma_{+} S
    \label{eq:qtritBnd}
\end{equation}
where $\alpha \in [0,5/2]$ and $\ket{\psi_{+}} = \frac{1}{\sqrt{3}}(\ket{00} +
\ket{11} + \ket{22})$, $\sigma_{+} = \frac{1}{\sqrt{3}}(\dya{01} + \dya{12} +
\dya{20})$, $S$ is the SWAP operator, $W$ in~\eqref{eq:swpd}, with $d=3$.  This
state is known to be entangled for $\alpha \in[0,2)$ and separable for $\alpha
\in [2,5/2]$.  Entanglement can be detected via the PPT criterion only for
$\alpha \in [0,1)$. For other values $1 \leq \al <2$ the state is both PPT and
entangled. 

\bigskip
 
\noindent 
{\bf Collection of $3 \times 3$ states with PPT entanglement.}
Let $d_a=3,d_b=3$ and consider the family of states described
in~\cite{Horodecki97}:

\begin{equation}
         \rho = \frac{1}{8y+1}
    \begin{pmatrix}
        A & B & C\\
        B^{\dagger} & A & D \\
        C^{\dagger} & D^{\dagger} & E
    \end{pmatrix},
    \label{eq:3x3bnd}
\end{equation}
where $y \in [0,1]$, $A = \text{diag}(y,y,y)$ and
$$
B  = \begin{pmatrix}
        0 & y & 0\\
        0 & 0 & 0 \\
        0 & 0 & 0
    \end{pmatrix},
        C = \begin{pmatrix}
        0 & 0 & y\\
        0 & 0 & 0 \\
        0 & 0 & 0
    \end{pmatrix},
        D = \begin{pmatrix}
        0 & 0 & 0\\
        0 & 0 & y \\
        0 & 0 & 0
    \end{pmatrix},
            \quad \text{and} \quad
        E = \frac{1}{2}
        \begin{pmatrix}
        1+y & 0 & \sqrt{1-y^2}\\
        0 & 2y & 0 \\
        \sqrt{1-y^2} & 0 & 1+y
    \end{pmatrix}.
$$
The state $\rho$ is entangled for $y \in(0,1]$ and separable for $y=0$. This
state is PPT for all $0 \leq y \leq 1$ and thus entanglement cannot be detected
via the PPT criterion.  

\bigskip 
 
\noindent 
{\bf Collection of $2 \times 4$ states with PPT entanglement.}
Let $d_a=2,d_b=4$ and consider the family of states described in~\cite{Horodecki97}:
\begin{equation}\label{eq:2x4}
   \rho = \frac{1}{7x+1}
    \begin{pmatrix}
        x & 0 & 0 & 0 & 0 & x & 0 & 0 \\
        0 & x & 0 & 0 & 0 & 0 & x & 0 \\ 
        0 & 0 & x & 0 & 0 & 0 & 0 & x \\
        0 & 0 & 0 & x & 0 & 0 & 0 & 0 \\
        0 & 0 & 0 & 0 & (1+x)/2 & 0 & 0 & \sqrt{1-x^2}/2\\
        x & 0 & 0 & 0 & 0 & x & 0 & 0\\
        0 & x & 0 & 0 & 0 & 0 & x & 0 \\
        0 & 0 & x & 0 & \sqrt{1-x^2}/2 & 0 & 0 &(1+x)/2
    \end{pmatrix},
    \end{equation}
where $x \in [0,1]$.  The state $\rho$ is entangled for $x\in (0,1)$ and
separable for $x=0$ and for $x=1$.  This state is PPT for all $x\in[0,1]$ and
thus entanglement cannot be detected via the PPT criterion. 

\bigskip

\subsection{Additional numerical results}

The next figures summarize numerical results analogous to those in  
Figures~\ref{figApstiso}--\ref{figApst3by3}, Figures~\ref{figextpstiso}--\ref{figextpst3by3}, and Figures~\ref{figcompareiso}--\ref{figcompare3by3} but for the Werner states and for the $2\times 4$ collection of states.

\begin{figure}[!ht]
 \hspace{-.6in}
\begin{tabular}{c} 
\begin{subfigure}{.6\textwidth}
    \centering
        \includegraphics[width=.8\textwidth]{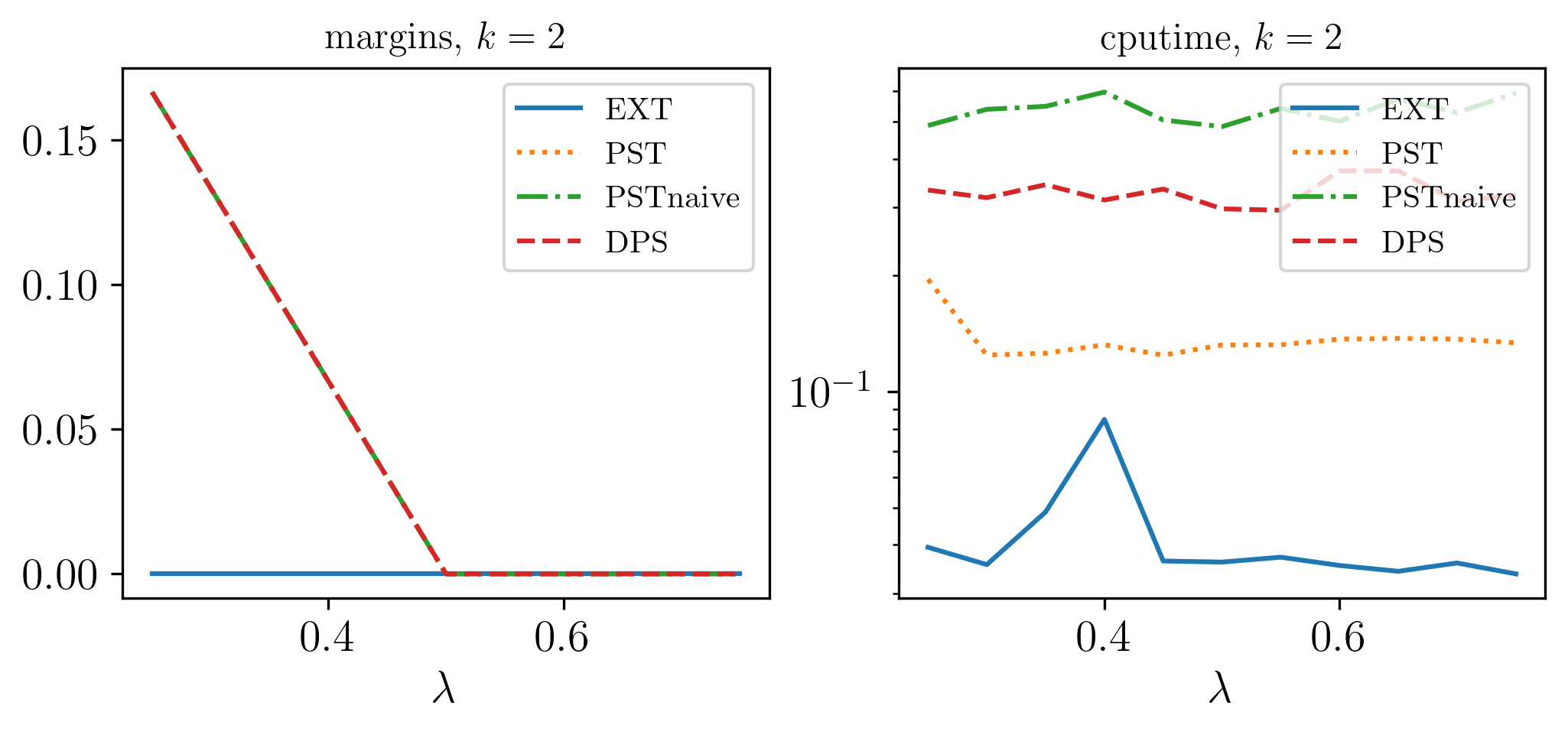}
        \caption{Werner states}
    \end{subfigure} 
    \hspace{-.4in}
    \begin{subfigure}{.6\textwidth}
    \centering
        \includegraphics[width=.8\textwidth]{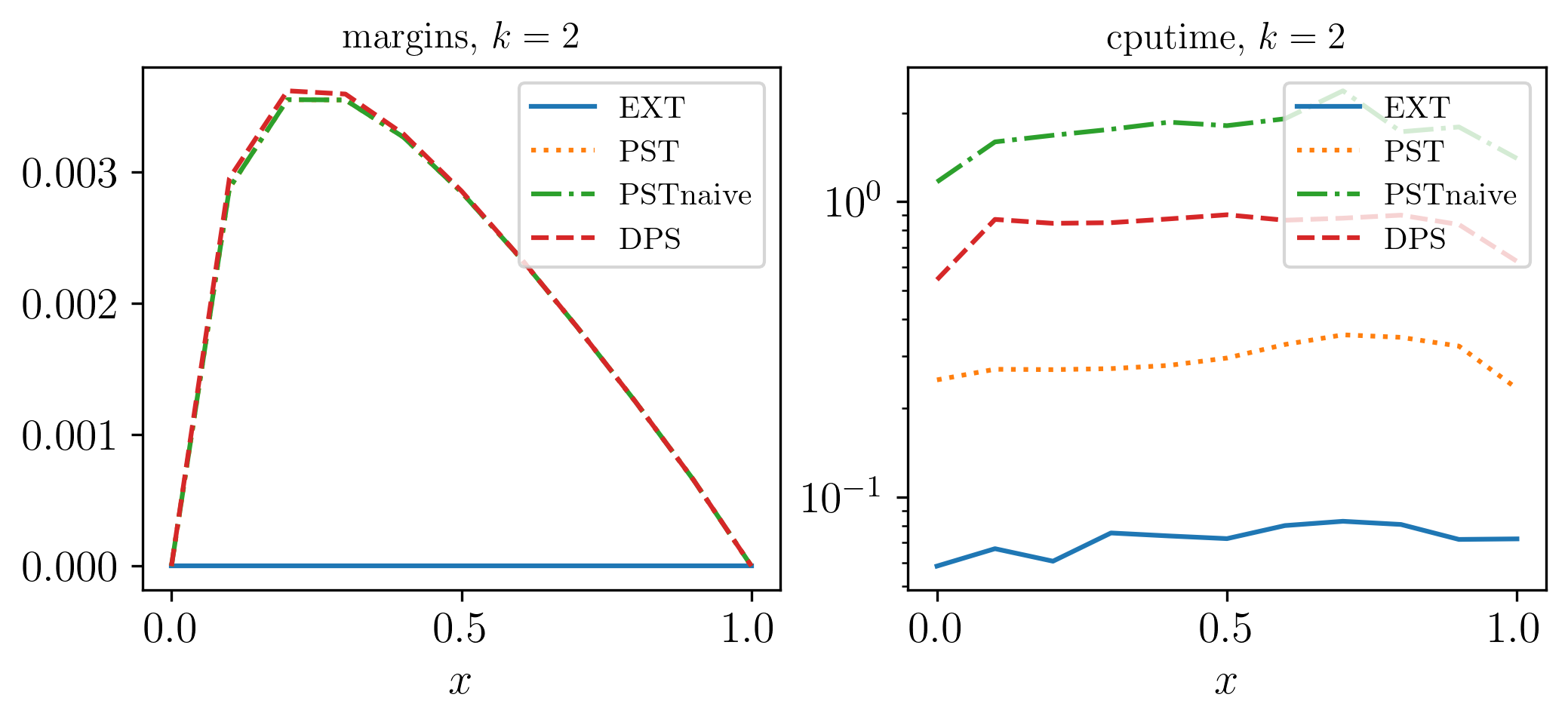}
        \caption{$2\times 4$ states}
    \end{subfigure} 
\end{tabular}
\caption{Comparison of EXT, PST, PST (naive), and DPS on the Werner and $2 \times 4$ collections.}\label{figApstwerner}
\end{figure}

\begin{figure}[!ht]
 \hspace{-.6in}
\begin{tabular}{c} 
\begin{subfigure}{.6\textwidth}
    \centering
        \includegraphics[width=.8\textwidth]{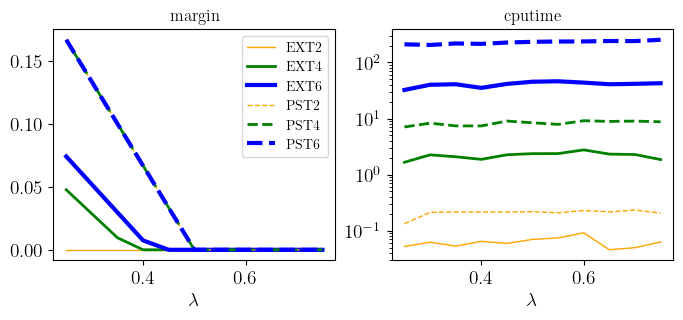}
        \caption{Werner states}
    \end{subfigure} 
    \hspace{-.4in}
    \begin{subfigure}{.6\textwidth}
    \centering
        \includegraphics[width=.8\textwidth]{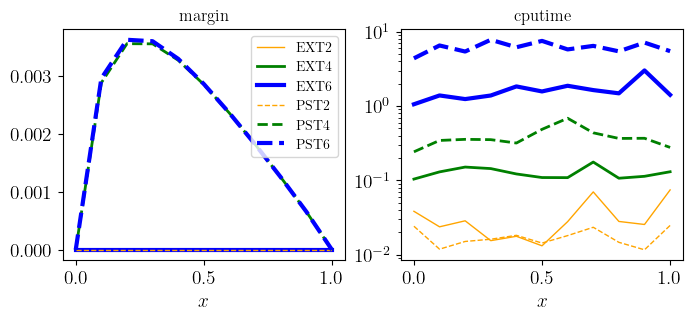}
        \caption{$2\times 4$ states}
    \end{subfigure} 
\end{tabular}
\caption{Comparison of EXT and PST on the Werner and $2\times 4$ collections.}\label{figextpstwerner}
\end{figure}

\begin{figure}[!ht]
 \hspace{-.6in}
\begin{tabular}{c} 
\begin{subfigure}{.6\textwidth}
    \centering
        \includegraphics[width=.8\textwidth]{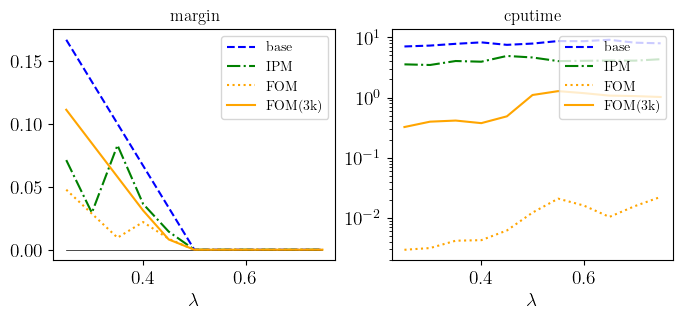}
        \caption{k = 4}
    \end{subfigure} 
    \hspace{-.4in}
    \begin{subfigure}{.6\textwidth}
    \centering
        \includegraphics[width=.8\textwidth]{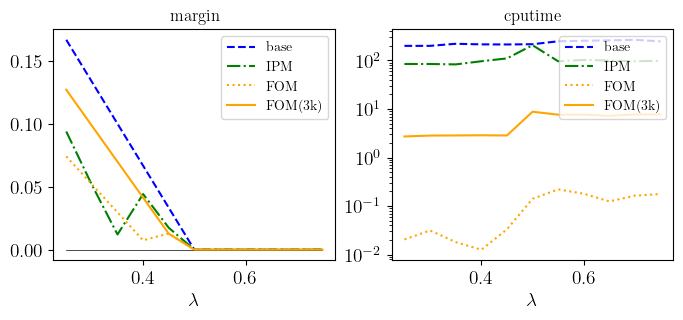}
        \caption{k = 6}
    \end{subfigure} 
\end{tabular}
\caption{Comparison of FOM, IPM, and baseline for PST on the Werner collection.}\label{figcomparewerner}
\end{figure}

\begin{figure}[!ht]
 \hspace{-.6in}
\begin{tabular}{c} 
\begin{subfigure}{.6\textwidth}
    \centering
        \includegraphics[width=.8\textwidth]{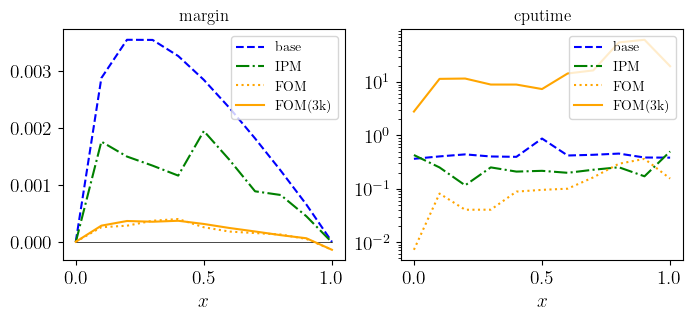}
        \caption{k = 2}
    \end{subfigure} 
    \hspace{-.4in}
    \begin{subfigure}{.6\textwidth}
    \centering
        \includegraphics[width=.8\textwidth]{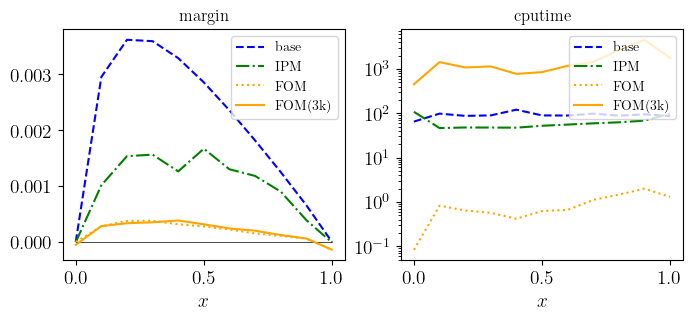}
        \caption{k = 4}
    \end{subfigure} 
\end{tabular}
\caption{Comparison of FOM, IPM, and baseline for PST on the $2\times 4$ collection.}\label{figcompare2by4}
\end{figure}

\subsection{Pathological instances and a natural preconditioning antidote}

The above examples of entangled states are all detectable via $\PST_2$ with the IPM and base configurations.  As it is discussed in \cite{DohertyParriloEA04},  for all $k=2,3,\dots$ there is a canonical construction that transforms an entangled and PPT state into another entangled state that belongs to $\DPS_k$ and thus also to $\PST_k$.  The transformation is as follows: for $\gamma > 0$ let $D_\gamma := \text{diag}(1,\gamma,\dots,\gamma)$ and
$$
\rho_\gamma = c(I_a\otimes D_\gamma) \rho (I_a\otimes D_\gamma)
$$
where $c>0$ is a normalization constant such that $\text{trace}(\rho_{\gamma})=1$.
Our numerical experiments confirm that indeed this transformation may convert entangled states into states that are not detectable via $\PST_k$.  \blue{More precisely, if $\rho$ is the qutrit state with $\alpha = 1.9$ and $\gamma = 0.3$ then the entanglement of $\rho_\gamma$ is not detectable via $\PST_2$ or even the tighter $\DPS_2$.  The entanglement of $\rho_\gamma$ is detectable via both $\PST_3$ or $\DPS_3$.  Detection via $\PST_3$ is substantially faster than via $\DPS_3$ (one or two seconds versus several minutes).}


The above transformation suggests the following simple {\em preconditioning} procedure. Suppose $\rho\in \HC_{ab}$ is a density operator and we want to detect if $\rho$ is entangled.  First, compute $\rho_b = \Tr_a(\rho)$.  We can assume $\rho_b$ is full rank and thus non-singular as otherwise the problem can be reduced to a lower dimensional subspace of $\HC_b$.  
Under that assumption, consider the {\em preconditioned state}
$$
\bar \rho :=\frac{1}{d_b}(I_a\otimes \rho_b^{-1/2}) \rho (I_a\otimes \rho_b^{-1/2})
$$
It is easy to see that $\rho$ is entangled if and only if the transformed density operator $\bar \rho$ is entangled.
This preconditioning procedure rules out examples of the form $\rho_\gamma$ and naturally raises the following research task: construct explicit instances of entangled states $\rho\in \HC_{ab}$ that satisfy both $\Tr_a(\rho)=I_b/d_b$ and $\rho \in \PST_k$.  

\blue{
For the five collections we have tested, this preconditioning procedure only changes the states in the $3\times 3$ and $2\times 4$ collections and the change is small.  Figure~\ref{fig.precond.3by3} 
illustrates the effect of preconditioning on margin and CPU time for $\PST_3$ on these two  collections of states.  The thicker lines display results after preconditioning.  Although preconditioning involves only a slight modification for states in this collection, the results show that preconditioning indeed yield some improvement in margin and running time.
}

\begin{figure}[!ht]
 \hspace{-.6in}
\begin{tabular}{cc}
\begin{subfigure}{.6\textwidth}
    \centering
        \includegraphics[width=.8\textwidth]{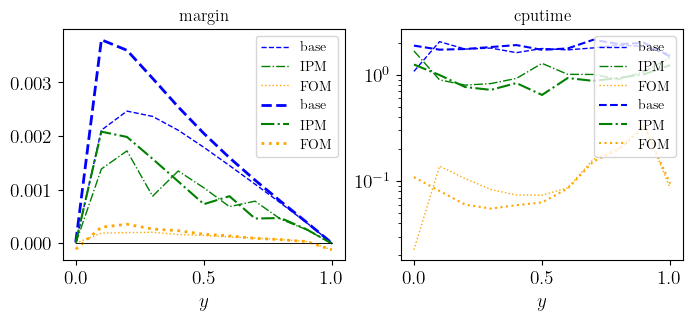}
        \caption{$3 \times 3$ collection}
    \end{subfigure} 
    \hspace{-.4in}
    \begin{subfigure}{.6\textwidth}
    \centering
        \includegraphics[width=.8\textwidth]{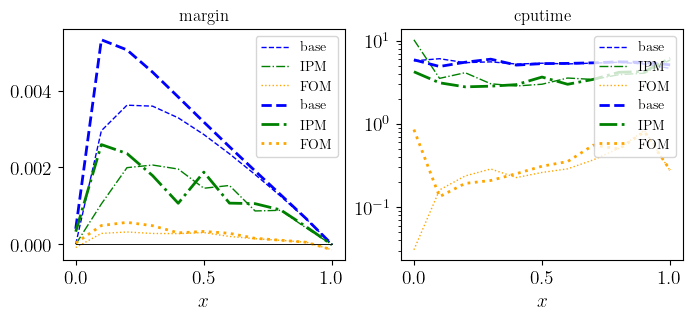}
        \caption{$2 \times 4$ collection}
    \end{subfigure}     
\end{tabular}
\caption{Effect of preconditioning on FOM, IPM, and baseline for $\PST_3$.}\label{fig.precond.3by3}
\end{figure}

\end{document}